\documentclass[a4paper,11pt]{article}
\usepackage{upgreek,tipa,cancel,mathrsfs,mathtools,mdframed}
\usepackage{amssymb,amsfonts,amsmath, amscd, amsthm, amsfonts,bm,fontawesome,tcolorbox,xcolor,quiver}
\usepackage{jheppub} 

\usepackage[T1]{fontenc} 

\newtheorem{thm}{Theorem}
\newtheorem{lemma}[thm]{Lemma}
\newcommand{\dd}{\mathrm{d}}
\newcommand{\UHP}{\mathbb{H}}
\newcommand{\Hilbert}{\mathcal{H}}

\newcommand{\cmpx}{\mathbb{C}}

\newcommand{\schottky}{\mathfrak{S}}

\newcommand{\linebundle}{\mathscr{L}}

\newcommand{\PSLC}{\operatorname{PSL}(2,\mathbb{C})}

\newcommand{\Gpotential}{\mathscr{S}}
\newcommand{\Lponetial}{\mathsf{h}}

\newcommand{\SchottkyFund}{\mathcal{D}}

\newcommand{\sing}{\operatorname{Sing}}

\newcommand{\orderrange}{\hat{\mathbb{N}}^{^{>1}}}

\newcommand{\brdiv}{\mathscr{D}}
\newcommand{\singrigon}{\overset{{}_{\curlywedge}}{\Omega}}
\newcommand{\singfund}{\overset{{}_{\curlywedge}}{\mathcal{D}}}

\newcommand\smallO{
	\mathchoice
	{{\scriptstyle\mathcal{O}}}
	{{\scriptstyle\mathcal{O}}}
	{{\scriptscriptstyle\mathcal{O}}}
	{\scalebox{.7}{$\scriptscriptstyle\mathcal{O}$}}
}
\newcommand{\stks}[1]{
	\left< #1 \right>
}
\newcommand{\compfontuii}[1]{\bm{\mathsf{#1}}^{\bullet , \bullet}}
\newcommand{\compfontui}[1]{\bm{\mathsf{#1}}^{\bullet}}
\newcommand{\compfontdii}[1]{\bm{\mathsf{#1}}_{\bullet , \bullet}}
\newcommand{\compfontdi}[1]{\bm{\mathsf{#1}}_{\bullet}}
\newcommand{\compfont}[1]{\bm{\mathsf{#1}}}

\title{\boldmath 
	Renormalized Volume, Polyakov Anomaly, and Orbifold Riemann Surfaces}

\author[a,b]{Hossein Mohammadi}
\author[c]{,Ali Naseh}
\author[c]{,Behrad Taghavi}

\affiliation[a]{Department of Physics, Sharif University of Technology, P.O. Box 11155-9161, Tehran, Iran}
\affiliation[b]{
	Research Center for High Energy Physics
	Department of Physics, Sharif University of Technology,
	P.O.Box 11155-9161, Tehran, Iran
}
\affiliation[c]{School of Particles and Accelerators, Institute for Research in Fundamental Sciences (IPM),
	P.O. Box 19395-5531, Tehran, Iran}

\emailAdd{hossein\_mohammadi@physics.sharif.edu}
\emailAdd{naseh,btaghavi@ipm.ir}

\abstract{In \cite{Taghavi2024classical}, two of the authors studied the function $\Gpotential_{\boldsymbol{m}} = S_{\boldsymbol{m}} - \pi \sum_{i=1}^n (m_i - \tfrac{1}{m_i}) \log \Lponetial_{i}$ for orbifold Riemann surfaces of signature $(g;m_1,...,m_{n_e};n_p)$ on the generalized Schottky space $\schottky_{g,n}(\boldsymbol{m})$. In this paper, we prove the holographic duality between $\Gpotential_{\boldsymbol{m}}$ and the renormalized hyperbolic volume $V_{\text{ren}}$ of the corresponding Schottky 3-orbifolds with lines of conical singularity that reach the conformal boundary. In case of the classical Liouville action on $\schottky_{g}$ and $\schottky_{g,n}(\boldsymbol{\infty})$, the holography principle was proved in \cite{krasnov2000holography} and \cite{park2017potentials}, respectively. Our result implies that $V_{\text{ren}}$  acts as a  K{\"a}hler potential for a particular combination of the Weil–Petersson and Takhtajan–Zograf metrics that appears in the local index theorem for orbifold Riemann surfaces \cite{takhtajan2019local}. Moreover, we demonstrate that under the conformal transformations, the change of function $\Gpotential_{\boldsymbol{m}}$ is equivalent to the Polyakov anomaly, which indicates that the function $\Gpotential_{\boldsymbol{m}}$ is a consistent height function with a unique hyperbolic solution. Consequently, the associated renormalized hyperbolic volume $V_{\text{ren}}$ also admits a Polyakov anomaly formula. The method we used to establish this equivalence may provide an alternative approach to derive the renormalized Polyakov anomaly for Riemann surfaces with punctures (cusps), as described in \cite{albin2013ricci}.}

\begin{document} 
	\maketitle
	\flushbottom
	
	\section{Introduction}\label{sec:intro}
	In this paper, we explore the holographic duality for three-dimensional spaces with a conformal boundary given by an orbifold Riemann surface of signature $\left(g;m_1,\dots, m_{n_e};n_p\right)$. Such orbifold Riemann surfaces are  two-dimensional surfaces with $n_e$ conical points labeled by $n=1,...,n_e$, each with a ramification index $m_1,...,m_{n_e}$, and  $n_p$ punctures.\footnote{See \cite{Taghavi2024classical} for more details.} These spaces are constructed from Euclidean AdS$_3 (\equiv\hspace{-1mm}\mathbb{U}^3)$ through discrete identifications by Kleinian groups, a special class of which are Schottky (extended) groups. Employing the semi-classical approximation to the gravity path integral, we compute the appropriately renormalized volume of the space, corresponding to the gravitational action. As we will demonstrate, this renormalized volume precisely matches the function (generalized Liouville action) $\Gpotential_{\boldsymbol{m}}$ introduced in \cite{Taghavi2024classical}. Thus, this paper offers a new three-dimensional perspective on certain results previously obtained by two of the authors through purely two-dimensional analyses in \cite{Taghavi2024classical}. 
	
	Let us recall that for compact Riemann surfaces\footnote{We understand that surfaces with conical singularities are compact. However, in this manuscript, when we use the term "compact," we are specifically referring to surfaces that do not possess conical singularities.} of arbitrary genus, the relation between bulk renormalized volume and standard Liouville action was proven by Krasnov \cite{krasnov2000holography} and Takhtajan and Teo \cite{Takhtajan_2003}.\footnote{The recognition that the standard Liouville action emerges as the effective action on the boundary was noted in \cite{seiberg1999d1,skenderis2000quantum}.} Park, Tahtajan, and Teo \cite{park2017potentials} extended the holographic correspondence for punctured Riemann surfaces using both Schottky and quasi-Fuchsian global coordinates. More recently, Park and Teo \cite{park2018liouville} extended the results of \cite{park2017potentials} to orbifold Riemann surfaces with quasi-Fuchsian global coordinates; in particular, they prove the holographic correspondence for quasi-Fuchsian 3-orbifolds. From a physics perspective, these orbifolds were also examined by Chandra, Collier, Hartman, and Maloney \cite{chandra2022semiclassical}. However, a thorough investigation of the holographic correspondence for Riemann orbisurfaces with Schottky global coordinates has not yet been published, and in this paper, we fill this gap.\footnote{Several studies have investigated the Einstein-Hilbert action on AdS$_3$ with conical singularities in connection with Liouville vertex operator correlation functions (see, e.g., \cite{krasnov2002lambda,abajian2024correlation}). See also \cite{He:2024xbi}.}
	
	Beyond extending previous works \cite{park2017potentials,Takhtajan_2003}, there are also significant motivations stemming from the importance of exploring quantum gravity in three dimensions. Previous research \cite{maloney2010quantum,keller2015poincare} has shown that using only smooth saddle points to calculate the gravitational path integral in three-dimensional gravity leads to two main issues with the resulting regularized partition function. First, the twist range at a constant spin is continuous rather than discrete. Second, at high spins and energies near the edge of the spectrum, the density of states becomes negative. The first issue might be mitigated by considering recent findings that suggest the dual theory of AdS gravity corresponds to an ensemble of quantum systems \cite{saad2019jt,stanford2019jt,jafferis20243d}. To address the non-unitarity problem, it has been proposed to modify the three-dimensional theory by introducing massive particles, which implies the need to consider three-dimensional conical manifolds alongside smooth saddles, as discussed in \cite{benjamin2020pure}.\footnote{In the classical limit, the deficit angle $\Delta\phi$ sourced by a massive scalar
		particle is related to its mass by $\Delta\phi\approx 8\pi G_{N}m$, see \cite{deser1984three}.} Therefore, exploring the connection between this modified three-dimensional gravitational theory and the  $\Gpotential_{\boldsymbol{m}}$ on its conformal boundary could provide deeper insights into understanding consistent three-dimensional quantum gravity.

	Another key motivation arises from the study of black hole formation by point particles \cite{Krasnov_2001,krasnov2002lambda}: in the collision of $n$ point particles, the resulting black hole can have a non-trivial topology inside the event horizon. Since the geometry inside the horizon is inaccessible to an external observer, it can be interpreted as a manifestation of black hole microstates. The results of Ref.\hspace{.5mm}\cite{krasnov2002lambda} demonstrate that, for each internal topology, the probability of black hole production in such collisions is determined by the $2n$-point function of heavy Liouville vertex operators on the asymptotic boundary of spacetime, represented by the Schottky double of a spatial slice. In the semi-classical limit, this probability is dominated by the exponential of the classical Liouville action evaluated on a genus-$g$\footnote{Note that $g \geq 2$ when the internal topology of the black hole is non-trivial.} Riemann surface with conical singularities at the insertion points of the vertex operators.
		
	Following the perspective that the bulk renormalized volume $V_{\text{ren}}$ is connected to the function $\Gpotential_{\boldsymbol{m}}$, in Section \ref{Section3} we show that this volume satisfies  Polyakov-type anomaly formula. To establish this, we first examine the connection between the $\Gpotential_{\boldsymbol{m}}$ on orbifold Riemann surfaces and Polyakov anomaly. The key idea underlying this connection is the uniformization theorem.
	
	The uniformization theorem states that every simply connected Riemann surface can be conformally mapped to one of three canonical geometries: the sphere (positive curvature), the Euclidean plane (zero curvature), or the hyperbolic plane (negative curvature). For surfaces of genus greater than 1, the theorem implies that the surface admits a unique hyperbolic metric within its conformal class. The Liouville action offers one method to find this unique hyperbolic metric, which provides a variational principle where the critical point corresponds to that unique hyperbolic metric. Another method involves solving the extremal problem for $-\log\det\Delta$.\footnote{The approach of proving the uniformization theorem by solving an extremal problem within a conformal class of metrics dates back at least to Berger's work \cite{berger1971riemannian}.}
	
	Let us recall that the modern quantum geometry of strings primarily explores all surfaces, analyzing variations in their metrics and the determinants of the corresponding Laplacians. A key area of interest is how this determinant, viewed as a function of the metric on a given surface, behaves, particularly in identifying its extreme values under specific metric constraints. This topic is extensively examined in the seminal work by Osgood, Phillips, and Sarnak \cite{osgood1988extremals}. They investigated the function $-\log\det\Delta$\footnote{The zero eigenvalue(s) are subtracted.} as a height function on the space of metrics for a compact, orientable, smooth surface of genus $g$. They discovered that for surfaces with $g> 1$, this function reaches its minimum at the unique (up to scaling and treating isometric surfaces as equivalent) hyperbolic metric within any given conformal class of metrics and has no other critical points.

	This highlights a deep connection between the spectrum of $-\log\det\Delta$ (and thus the Polyakov anomaly)  and the Liouville action, identified by Takhtajan and Teo for compact Riemann surfaces, see \cite{Takhtajan_2003}. We show that this connection can be extended to orbifold Riemann surfaces using $\Gpotential_{\boldsymbol{m}}$. The method we employed to find this extension may provide an alternative approach for deriving the renormalized Polyakov anomaly for Riemann surfaces with punctures (cusps), whether or not they have conical singularities. Accordingly, since we have demonstrated that the renormalized hyperbolic volume $V_{\text{ren}}$ coincides with $\Gpotential_{\boldsymbol{m}}$, it follows that this volume also obeys Polyakov anomaly.
	
	The structure of this paper is as follows: In Section~\ref{Renvolume}, we compute the bulk renormalized volume $V_{\text{ren}}$ using the framework of double (co)homology complexes. Section~\ref{Section3} demonstrates that $V_{\text{ren}}$ satisfies the Polyakov anomaly. Finally, in Section~\ref{Conclusion}, we present our conclusions and discuss some potential future research directions. The mathematical background for this paper, including the definitions of group homology and group cohomology, is primarily drawn from \cite{Aldrovandi_1997}. Nearly all the essential mathematical details are presented in Appendices \ref{asymapp}-\ref{regsurf} to make the paper self-contained.
	\section{Renormalized Volume and Holography Principle}\label{Renvolume}
	In bulk three dimensions, the on-shell value of the Einstein-Hilbert action is expected to be proportional to the hyperbolic volume of the 3-manifold $M$. However, since the metric diverges at the conformal boundary, the volume is infinite, and one needs to regularize it by truncating the three-dimensional manifold $M$ by a cutoff surface labeled by a parameter $\varepsilon$. Let us denote this surface by $f= \varepsilon$. In asymptotically (locally) AdS spaces, the volume grows as half of the boundary area. Accordingly, to do the regularization procedure, one should calculate the volume $V_{\varepsilon}$ above the cutoff surface of fixed $\varepsilon$, then subtract $A_{\varepsilon}/2$ and finally take the limit $\varepsilon \rightarrow 0$. It is well-known that subtracting only the area term does not eliminate all the divergences, leaving a logarithmic divergence that must be addressed. By subtracting this logarithmic divergence, the final result becomes scale-dependent, which is linked to the conformal anomaly of the boundary dual theory.\textcolor{blue}{\footnote{Early examples of renormalized volumes can be found in the work of Henningson and Skenderis \cite{henningson1998holographic} for asymptotically hyperbolic Einstein metrics.}} In this paper, we compute the renormalized volume of a three-dimensional Schottky manifold—commonly known as a handlebody geometry—with lines of conical singularities that reach the conformal boundary. This manifold is $M=(\mathbb{U}^3\cup\singrigon)\slash\Sigma$ with the conformal boundary at infinity given by $X=\singrigon\slash\Sigma$, where $\Sigma$ is marked normalized Schottky group with region of discontinuity $\singrigon\subset \mathbb{C}$ (see Appendix \ref{Schotkkyreview} for more details).
	
	In Appendix \ref{regsurf}, we provide a detailed explanation of constructing the regularizing surface $f=\varepsilon$\hspace{1mm} for Schottky manifolds, whose conformal boundaries may or may not contain conical singularities and punctures. To offer a general overview, we present a simplified (intuitive) illustration of the compact case. However, interested readers should refer to Appendix D for a more comprehensive treatment.
	
	For the compact case,\footnote{See Appendix \ref{Schotkkyreview}.} the Schottky manifold  $M$  is formed by taking the quotient of $\mathbb{U}^3\cup \Omega$ by the Schottky group $\Sigma$. Consequently, the appropriate regularizing surface on  $M$ must be a $\Sigma$-automorphic function. The group $\Sigma$ is a loxodromic subgroup of  $\PSLC$, the isometry group of Euclidean AdS$_{3}$ space. We first determine this isometry group and then use it to identify the suitable regularizing surface $f=\varepsilon$. For our purposes, it is most convenient to work with $\text{AdS}_3$ space in Poincare coordinate. The hyperbolic metric in this coordinate is 
	\begin{equation}
	\begin{aligned}
	ds^2_{\text{\tiny AdS}_3} = \frac{dr^2 + dz d\bar{z}}{r^2},
	\end{aligned}\label{AdS3}
	\end{equation}
	where the conformal boundary is located at $r=0$ and $z$ is a holomorphic coordinate on it. The isometry group of the metric \eqref{AdS3} is $\PSLC$ and let denote it by $\Gamma$. To construct the $\Gamma$-invariant cutoff surface $f$, one needs the explicit form of the $\PSLC$-action on $\mathbb{U}^3$, which is defined by utilizing quaternionic numbers. Let us represent the $Z =(z,r) \in \mathbb{U}^3$ by a quaternion 
	\begin{equation}
	\mathbf{Z} = x.\mathbf{1} + y.\mathbf{i}+r.\mathbf{j},
	\end{equation}
	where 
	\begin{equation}
	\mathbf{1} = \begin{pmatrix}
	1 & 0 \\
	0 & 1 
	\end{pmatrix},~~\mathbf{i} = \begin{pmatrix}
	i & 0 \\
	0 & -i 
	\end{pmatrix},~~\mathbf{j} = \begin{pmatrix}
	0 & -1 \\
	1 & 0 
	\end{pmatrix}.
	\end{equation}
	Accordingly, the complex variable $z\in \mathbb{C}$ is represented by
	\begin{equation}
	\mathbf{z} = \Re(z).\mathbf{1}+\Im(z).\mathbf{i} 
	\end{equation}
	The transformation of $Z$ under the group element $\gamma = \begin{pmatrix}
	a & b \\
	c & d 
	\end{pmatrix}\in \PSLC$, i.e. $Z\rightarrow \gamma Z$, is given by
	\begin{equation}
	\mathbf{Z} \rightarrow  \frac{\mathbf{a}~\mathbf{Z} + \mathbf{b}}{\mathbf{c}~\mathbf{Z}+\mathbf{d}}
	\end{equation}
	Explicitly, we have
	\begin{equation}
	\mathbf{a}~\mathbf{Z} + \mathbf{b} = \begin{pmatrix}
	a & 0 \\
	0 & \bar{a}
	\end{pmatrix}\begin{pmatrix}
	z & -r \\
	r & \bar{z} 
	\end{pmatrix}+\begin{pmatrix}
	b & 0 \\
	0 & \bar{b}
	\end{pmatrix} = \begin{pmatrix}
	a z+b & -a r \\
	\bar{a} r & \bar{a} \bar{z}+\bar{b}
	\end{pmatrix},
	\end{equation}
	so that
	\begin{equation}\label{gZ}
	\gamma Z = \begin{pmatrix}
	a z+b & -a r \\
	\bar{a} r & \bar{a} \bar{z}+\bar{b}
	\end{pmatrix}\begin{pmatrix}
	c z+d & -c r \\
	\bar{c} r & \bar{c} \bar{z}+\bar{d}
	\end{pmatrix}^{-1}.
	\end{equation}
	Since
	\begin{equation}
	\begin{pmatrix}
	c z+d & -c r \\
	\bar{c} r & \bar{c} \bar{z}+\bar{d}
	\end{pmatrix}^{-1} = J_{\gamma}(Z) \begin{pmatrix}
	\bar{c} \bar{z}+\bar{d} & c r \\
	-\bar{c} r & c z + d
	\end{pmatrix},
	\end{equation}
	where 
	\begin{equation}
	J_{\gamma}(Z) =1/\left(|c z +d|^2+|c r|^2\right),\label{Jgammaz}
	\end{equation}
	one can see that the equation \eqref{gZ} can be written as
	\begin{equation}
	\gamma Z = J_{\gamma}(Z) \begin{pmatrix}
	(a z +b)(\overline{c z +d})+ a \bar{c}~r^2 ~~~& c r (a z +b)- a r (c z +d) \\
	\bar{a} r (\overline{c z +d})-\bar{c} r(\overline{a z +b}) & (c z + d)(\overline{a z +b})+ \bar{a} c~r^2
	\end{pmatrix}.
	\end{equation}
	Setting $z(Z)= z$, $r(Z)=r$ and $\gamma Z= \begin{pmatrix}
	z(\gamma Z) & -r(\gamma Z) \\
	r(\gamma Z) & \bar{z}(\gamma Z)
	\end{pmatrix}$, together with $ad -bc=1$, gives
	\begin{equation}
	\begin{aligned}
	z(\gamma Z) &= J_{\gamma}(Z)\left((a z +b)(\overline{c z +d})+ a \bar{c}~r^2\right),\\\bar{z}(\gamma Z) &= J_{\gamma}(Z) \left((c z + d)(\overline{a z +b})+ \bar{a} c~r^2\right),\\
	r(\gamma Z) &= J_{\gamma}(Z)~r.
	\end{aligned}\label{adstr}
	\end{equation}
	The transformations in equation \eqref{adstr} are simply the isometries of the metric \eqref{AdS3}. For the regularization procedure discussed, we require the near-boundary ($r\rightarrow 0$) form of these transformations, which are as follows
	\begin{equation}
	\begin{aligned}
	z(\gamma Z) &= \gamma(z)+\mathcal{O}(r^2),\\\bar{z}(\gamma Z) &= \overline{\gamma(z)}+\mathcal{O}(r^2),\\
	r(\gamma Z) &= |\gamma'|~r+ \mathcal{O}(r^3).
	\end{aligned}\label{asyadstr}
	\end{equation}
	Since for every $\varphi \in \mathcal{C}\mathcal{M}(\Omega/\Sigma)$ we have $\varphi(\gamma(z))=\varphi(z)-\log|\gamma'(z)|^2$, the $\Sigma$-automorphic function\footnote{In deriving the asymptotic transformations \eqref{asyadstr}, we have just applied the condition $\text{det}\gamma=1$.} $f$ which is positive on $\mathbb{U}^3$ satisfies
	\begin{equation}
	f(Z)= r~e^{\varphi(z)/2} +\mathcal{O}(r^3)~~~~~ \text{as} ~~r\rightarrow 0.\label{fcompactt}
	\end{equation}
	In this simplified version, the explicit forms of subleading terms cannot be obtained rigorously. As mentioned earlier, Appendix \ref{regsurf} provides a comprehensive and detailed analysis. According to equation \eqref{fZ}, the precise form of the regularization surface for the compact case is given by
	\begin{equation}
	f(Z) = re^{\varphi(z)/2}-\frac{1}{2}e^{\varphi(z)/2}\varphi'(z)\left(\sum_{\gamma \in \Gamma} \eta(\gamma Z)\frac{\overline{\gamma''(z)}}{\overline{\gamma'(z)}}\hspace{1mm} \right)r^3+\mathcal{O}(r^5).\label{fcompact}
	\end{equation} 
	As explained in the Appendix \ref{regsurf}, this form remains applicable even in the presence of conical singularities and punctures, which are obtained as quotients of the symmetric group acting on $\mathbb{U}^{3}\cup\Omega$. In our case, we aim to define the regularizing surface $f=\varepsilon$ by truncating the regions around the conical singularities and punctures. To accomplish this, it suffices to use the expression in \eqref{fZ} and substitute the asymptotic behavior of the field $\varphi$ near each singularity to determine the regularizing surface. Based on and utilizing \eqref{conicalasymp}, near each conical singularity and puncture, one finds\footnote{Appendix \ref{regsurf} discusses two types of regularization. The first method, based on the works \cite{kra1972} and \cite{Takhtajan_2003}, is used to derive equations \eqref{fcompact}, \eqref{fconical}, and \eqref{fpuncture}. The second regularization approach, employed to obtain equation \eqref{fcompactt}, is further extended in the last part of Appendix \ref{regsurf} to account for the presence of singularities.}
	\begin{equation}
	f(Z)_{\text{Conical}} = re^{\varphi(z)/2}+\mathcal{O}\left(\frac{r^3}{|z-z_j|^{2-\frac{1}{m_j}}}\right),\label{fconical}
	\end{equation} 
	and
	\begin{equation}
	f(Z)_{\text{Cusp}} = re^{\varphi(z)/2}+\mathcal{O}\left(\frac{r^3}{|z-z_j|^{2}\log|z-z_j|}\right).\label{fpuncture}
	\end{equation}
	Based on equations \eqref{fcompact}, \eqref{fconical}, and \eqref{fpuncture}, it is important to note that the leading behavior of the function $f$ remains the same whether or not singularities are present. Furthermore, the level surface $f=\varepsilon$ intersects the points $(z_i,0)$s, making it non-compact. Therefore, to use  $f$ as a level-defining function for the truncated fundamental region $R\cap\{f=\varepsilon\}$,\footnote{See Appendix.\ref{homology}.} it is necessary to remove neighborhoods around the points $(z_i,0)$ in $\mathbb{U}^3$. As a result, the regularized truncated fundamental region should be defined as follows:\footnote{The $||~||$ is the Euclidean distance in $\overline{\mathbb{U}}^3= \mathbb{U}^3\cup \mathbb{C}$.}
	\begin{equation}
	R_{\varepsilon} = R\cap\{f=\varepsilon\}\backslash \bigcup_{j=1}^{n}\big{\{}(z,r)\in\mathbb{U}^3\big| ~|| (z,r)-(z_i,0)|| \leq \tilde{\varepsilon}\big{\}}.\label{Repsilon1}
	\end{equation}
	To determine the parameter $\tilde\varepsilon$, one should observe that in the limit $r\rightarrow 0$, under Schottky transformations by each generator $L_k\in\Sigma$, we have
	\begin{equation}
	|L_k z-L_k z_j| = (z-z_j)L'_k(z_j)+\mathcal{O}(z-z_j)^2 < \tilde\varepsilon,
	\end{equation}
	implying that if $\tilde\varepsilon$ is assumed to be constant, then the region $|z - z_j| \leq \tilde\varepsilon$, up to $\mathcal{O}(\tilde\varepsilon^2)$ terms, transforms to
	\begin{equation}
	|z-z_j| \leq  \tilde\varepsilon/L'_k(z_j).    
	\end{equation}
	However, since the point $L_k z_j$ has the same branching number $m_i$, the radius of the regularized circle should remain unchanged, implying that the parameter $\tilde\varepsilon$ also changes under the Schottky transformations. The Hauptmodule coefficients $J_1^{(j)}$ in equations \eqref{Jparabolic1} and \eqref{Jparabolic2} transform as $J_1^{(j)} \rightarrow L'_k(z_j) J_1^{(j)}$ under Schottky transformations (see Section~3 in \cite{Taghavi2024classical}). This leads to $\tilde\varepsilon = \left|J_1^{(j)}\right| \varepsilon$, and the equation \eqref{Repsilon1} changes to
	\begin{equation}
	R_{\varepsilon} = R\cap\{f=\varepsilon\}\backslash \bigcup_{j=1}^{n}\big{\{}(z,r)\in\mathbb{U}^3\big| ~||(z,r)-(z_j,0)|| \leq \Big|J_{1}^{(j)}\Big|\hspace{.5mm}\varepsilon\big{\}}.\label{Repsilon}
	\end{equation}
	It is crucial to highlight that the cutoff surface $f(Z) = \varepsilon$ must adjust to the local geometry near the conical singularities, forming a shape that reflects the angular deficit, as outlined in equations \eqref{fconical} and \eqref{fpuncture}. While it deforms locally around the singularities,  the cutoff surface should globally align with the symmetry pattern imposed by the Schottky group (a subgroup of $\PSLC$) in smooth regions of the space, as indicated in equation \eqref{fcompact}. Additionally, during the subsequent calculations, it is important to remember that we first cut out the areas around the conical singularities and punctures and then uniformize the remaining surface using the standard Schottky group. Consequently, the group used for constructing the (co)homology double complexes will be the standard Schottky group, while the local information remains encoded in the behavior of the field 
	$\varphi$ near the singularities. If we had generated the conical singularities and punctures using the extended Schottky group, incorporating elliptic and parabolic elements, we would also need to develop double complexes for this extended group.

	Now, with the appropriate function $f$ determined, we can proceed to the detailed calculation of the renormalized volume of the corresponding three-dimensional Schottky manifold $M$ with lines of conical singularities and a conformal boundary $X= \singrigon\slash\Sigma$, where the fundamental domain of $\singrigon$ will be denoted by $\singfund_{\varepsilon}$.
	
	To compute the area term
	$A_{\varepsilon}$, we require the induced metric on the regularizing cutoff surface $f(Z)=\varepsilon$ which takes the following form:
	\begin{equation}
	\begin{aligned}
	ds^2 = \frac{f(Z)}{r^2} \bigg(
	\big(1+ \frac{f_x^2}{f_r^2}\big)dx^2 + 
	\big(1+ \frac{f_y^2}{f_r^2}\big)dy^2 +
	\frac{2f_xf_y}{f^2_r} dxdy
	\bigg).
	\end{aligned}
	\end{equation}
	So, the induced  area form on the regularizing surface becomes
	\begin{equation}
	\begin{aligned}
	d A_\varepsilon = \frac{1}{r^2} \sqrt{1+ \Big(\frac{f_x}{f_r}\Big)^2+ \Big(\frac{f_y}{f_r}\Big)^2} dx\wedge dy.
	\end{aligned}\label{area}
	\end{equation}
	By noting that
	\begin{equation}
	\begin{aligned}
	\frac{f_x(Z)}{f_r(Z)}= \frac{r}{2}~ \partial_x \varphi +\mathcal{O}(r^3), \qquad 
	\frac{f_y(Z)}{f_r(Z)}= \frac{r}{2}~ \partial_y \varphi + \mathcal{O}(r^3),
	\end{aligned}
	\end{equation}
	the area term \eqref{area} modifies to\footnote{In the second line, we transformed into complex coordinates.}
	\begin{equation}
	\begin{aligned}
	A_\varepsilon[\varphi] &= \iint_{\singfund_{\varepsilon}} \frac{dx\wedge dy}{r^2}\sqrt{1+ \frac{r^2}{4}(\partial_x^2 \varphi + \partial_y^2 \varphi)+ \mathcal{O}(r^4)}  
	~ + \smallO(1) \\ &= 
	\iint_{\singfund_{\varepsilon}}  \frac{dx\wedge dy}{r^2}  + \frac12\iint_{\singfund_{\varepsilon}}
	\partial_z \varphi \partial_{\bar{z}} \varphi \; dx\wedge dy + \smallO(1) \\ &=
	\iint_{\singfund_{\varepsilon}}  \frac{dx\wedge dy}{r^2} + \frac{i}{4} \stks{ \check{\omega}[\varphi],\singfund_{\varepsilon}},
	\end{aligned}
	\label{Aeps}
	\end{equation}
	where $\singfund_{\varepsilon}$ is the complement of the following region
	\begin{equation}
	\{f=\varepsilon\}\bigcap ~ \bigcup_{j=1}^{n}\big{\{} (z,r)\in \mathbb{U}^3 \big{|}\hspace{1mm}|(z,r)-(z_j,0)|\leq \Big|J_{1}^{(j)}\Big|\hspace{.5mm}\varepsilon\big{\}},\label{Depsi}
	\end{equation}
	and 
	\begin{equation}
	\check{\omega}[\varphi] = \omega[\varphi]-e^{\varphi} dz\wedge \dd\bar{z} = \partial_{z}\varphi \partial_{\bar{z}}\varphi~dz\wedge \dd\bar{z}.\label{omegacheck} 
	\end{equation}
	
	The regularized volume of the fundamental region is
	\begin{equation}
	\begin{aligned}
	V_\varepsilon[\varphi] = \stks{\omega_3,R_\varepsilon} &= \stks{\omega_3 , R_\varepsilon-S_\varepsilon}\\& \hspace{-2.5mm}
	\overset{\eqref{Dvarpi}}{=}\stks{D\varpi,R_{\varepsilon}-S_{\varepsilon}}\\& \hspace{-2.5mm}\overset{\eqref{varpi}}{=} \stks{D(\omega_2 -\omega_1-\omega_0) ,R_\varepsilon-S_\varepsilon } \\ &\hspace{-2.5mm}\overset{\eqref{Dpartial}}{=}
	\stks{\omega_2 -\omega_1-\omega_0, \partial(R_\varepsilon-S_\varepsilon)}\\& \hspace{-2.5mm}\overset{\eqref{RmSeps}}{=} \stks{\omega_2 -\omega_1-\omega_0, -\singfund_{\varepsilon}-L} \\ &=
	-\stks{\omega_2,\singfund_{\varepsilon}} + \stks{\omega_1,L}  
	\overset{\eqref{omega2}}{=} 
	\frac12 \iint_{\singfund_{\varepsilon}}  \frac{dx\wedge dy}{r^2} + \stks{\omega_1,L}.
	\end{aligned}\label{V1}
	\end{equation}
	In the first equality, we've used that adding a degree two element from the total homology complex does not change the integration since $\omega_3$ is a degree three element of the total cohomology complex.
	Also, in the last line, $\omega_0$ contribution vanishes due to the lack of vertices in the fundamental domain of Schottky uniformization. Now, let's focus on $\omega_1$ contribution. Since (see Appendix.\ref{CohomologyU3})
	\begin{equation}
	\begin{aligned}
	\big(\omega_1\big)_{\gamma^{-1}} = -\frac{i}{8} \log \Big(
	|r~c(\gamma)|^2 J_\gamma(Z)
	\Big) \Big(
	\frac{\gamma''}{\gamma'}dz -  \frac{\overline{\gamma''}}{\overline{\gamma'}}d\overline{z}
	\Big),
	\end{aligned}
	\end{equation}
	plugging from the cutoff function $f(Z) =  r~e^{\varphi(z)/2} + \mathcal{O}(r^3)=\varepsilon $ (see equations \eqref{fcompact}, \eqref{fconical} and \eqref{fpuncture}) and noting that 
	\begin{equation}
	J_{\gamma}(Z) = |\gamma'| +\mathcal{O}(r^2)~~~~\text{as}~~~~r\rightarrow 0
	\end{equation}
	results in
	\begin{equation}
	\begin{aligned}
	\big(\omega_1\big)_{\gamma^{-1}} &=
	-\frac{i}{8} \log \Big(
	|\varepsilon~c(\gamma)|^2 ~e^{-\varphi(z)} |\gamma'| 
	\Big) \Big(
	\frac{\gamma''}{\gamma'}dz -  \frac{\overline{\gamma''}}{\overline{\gamma'}}d\overline{z}
	\Big) +\smallO(1)\\ &= 
	-\frac{i}{8} \Big(
	2\log \varepsilon  -\varphi + \frac12 \log |\gamma'|^2 + \log |c(\gamma)|^2  
	\Big)\varkappa_{\gamma^{-1}}\\ &=
	-\frac{i}{4}  \varkappa_{ \gamma^{-1}}
	\log \varepsilon 
	+\frac{i}{8} ~\theta_{\gamma^{-1}} [\varphi],
	\end{aligned},\label{w1}
	\end{equation}
	where we have defined\footnote{It is worth noting that $\dd \chi_{\gamma^{-1}}=(\delta\chi)_{\gamma_1^{-1},\gamma_2^{-1}}=0.$}
	\begin{equation}
	\varkappa_{\gamma^{-1}} = \frac{\gamma''}{\gamma'}dz -  \frac{\overline{\gamma''}}{\overline{\gamma'}}d\overline{z} ~~~~~~~\text{and}~~~~~~\theta_{\gamma^{-1}}[\varphi] = \Big(\varphi - \frac12 \log |\gamma'|^2 - \log |c(\gamma)|^2  
	\Big)\varkappa_{\gamma^{-1}}.\label{chigm1}
	\end{equation}
	Substituting $\omega_{1}$ from \eqref{w1} into \eqref{V1} yields
	\begin{equation}
	\begin{aligned}
	V_\varepsilon[\varphi] =\frac12 \iint_{\singfund_{\varepsilon}} \frac{dx\wedge dy}{r^2} -\frac{i}{4}\stks{\varkappa_{\gamma^{-1}},L}\log \varepsilon + \frac{i}{8} \stks{\theta_{\gamma^{-1}}[\varphi],L}.
	\end{aligned}\label{V2}
	\end{equation}
	Moreover,
	\begin{equation}
	\begin{aligned}
	\stks{\varkappa_{\gamma^{-1}},L}&=\stks{\delta \varkappa_1,L} \\&  = \stks{\varkappa_1,\partial'' L} \overset{\eqref{parD}}{=} \stks{\varkappa_1,\partial' \singfund_{\varepsilon}} - \sum_{j=1}^{n}\stks{\varkappa_1, c'_j} \\ &=
	\stks{\dd\varkappa_1,\singfund_{\varepsilon}} - \sum_{j=1}^{n}\stks{\varkappa_1 ,  c'_j},\label{qq1}
	\end{aligned}
	\end{equation}
	with
	\begin{equation}
	\chi_{1} = \partial_{\bar{z}}\varphi ~\dd \bar{z}-\partial_{z}\varphi~\dd z.\label{chi1}
	\end{equation}
	The equality in the first line of \eqref{qq1} can be derived by observing that under the conformal transformation $z\rightarrow \tilde{z} =\gamma(z)$, we have
	\begin{equation}
	\dd z = \dd \tilde{z}/\gamma',~~~~\partial_{z}= \gamma' \partial_{\tilde{z}},~~~~~\tilde{\varphi}=\varphi(\gamma(z)) = \varphi-\log|\gamma'|^2,\nonumber
	\end{equation}
	Hence
	\begin{equation}
	(\delta\chi_1)_{\gamma_1^{-1}} = \chi_{1}.\gamma-\chi_{1} \overset{\eqref{chi1}}{=} -\partial_{z}\left(\varphi-\log|\gamma'|^2\right) \dd z +\partial_{\bar{z}} \left(\varphi-\log|\gamma'|^2\right)\dd \bar{z}-\chi_{1} \overset{\eqref{chigm1}}{=} \chi_{\gamma^{-1}}.
	\end{equation}
	Equation \eqref{qq1} can be further simplified by observing that $\dd \chi_{1} = 2\partial_{z}\partial_{\bar{z}}\varphi ~dz \wedge d\bar{z}$, and thus,
	\begin{equation}
	\stks{\dd \chi_{1},\singfund_{\varepsilon}} = 2\iint_{\singfund_{\varepsilon}} \partial_{z}\partial_{\bar{z}}\varphi~dz \wedge d\bar{z} = 2i\iint_{\singfund_{\varepsilon}} K[\varphi] ~e^{\varphi}d^2z,
	\end{equation}
	where $K[\varphi] = -2 e^{-\varphi} \partial_{z}\partial_{\bar{z}}\varphi$ is Gaussian curvature of the metric $e^{\varphi}\dd z\dd{\bar{z}}$. Accordingly,
	\begin{equation}
	\stks{\varkappa_{\gamma^{-1}},L}=2i~\Big(\iint_{\singfund_{\varepsilon}} K[\varphi] ~e^{\varphi}d^2z +\frac{i}{2}\sum_{j=1}^{n}\oint_{c_j^{'\hspace{.2mm}}}\left(\partial_{\bar{z}}\varphi ~\dd \bar{z}-\partial_{z}\varphi~\dd z\right)\Big).\label{qq2}
	\end{equation}
	Let us rewrite the boundary contributions in equation \eqref{qq2} using the geodesic curvature of each circular boundary $c_{j}^{'}$.\footnote{Recall that we are ultimately considering the  $r\rightarrow 0$ limit.} In the conformally flat metric $e^{\varphi}\dd z \dd \bar{z}$, if the boundary $c_{j}^{'}$ is a smooth curve, the geodesic curvature $k[\varphi]$ of $c_{j}^{'}$
	is given by
	\begin{equation}
	k[\varphi] =e^{-\varphi/2}\big(k[\text{flat}] +\partial_{n}\varphi\big),
	\end{equation}
	where $k[\text{flat}]$ is the geodesic curvature of the boundary with respect to the flat metric and $\partial_{n}\varphi$ is the normal derivative of $\varphi$ along the boundary. However, since no regularizing circle is needed for the flat metric, the contribution of $k[\text{flat}]$ should be excluded. Given a circular boundary with radius $\tilde{\varepsilon}$,\footnote{See the equation \eqref{Depsi}.} we have
	\begin{equation}
	\partial_{n}\varphi= e^{i\theta}\partial_{z}\varphi +e^{-i\theta}\partial_{\bar{z}}\varphi,~~~~\dd s=e^{\varphi/2}\tilde{\varepsilon}~\dd\theta= \frac{i}{2}e^{\varphi/2}\big(e^{i\theta}\dd\bar{z}-e^{-i\theta}\dd z\big),\label{k1}
	\end{equation}
	where $\dd s$ is the length element of the boundary in the conformally flat metric. Thus, with using asymptotic behaviour of the field $\varphi$ near the  singularities (see Appendix \ref{asymapp}), we obtain
	\begin{equation}
	\oint_{c_{j}^{'}} k[\varphi] ~\dd s =\frac{i}{2}\oint_{c_{j}^{'}}\left(\partial_{\bar{z}}\varphi ~\dd\bar{z}-\partial_{z}\varphi~ \dd z\right).\label{k2}
	\end{equation}
	Now, according to the Gauss-Bonnet theorem for a manifold $X$ with the metric $e^{\varphi}\dd z \dd \bar{z}$, i.e,
	\begin{equation}
	\iint_{X} K[\varphi] ~e^{\varphi}d^2z +\oint_{\partial X}k[\varphi]~ \dd s= 2\pi \chi(X),
	\end{equation}
	together with using \eqref{k2}, the equation \eqref{qq2} simplifies  to\footnote{For the orbifold Riemann surface with signature $\left(g;m_1,\dots, m_{n_e};n_p\right)$,  $\chi(X)=2-2g-\sum_{j=1}^{n_e}(1-\frac{1}{m_j})-n_p.$}
	\begin{equation}
	\stks{\varkappa_{\gamma^{-1}},L} = 4\pi i~ \chi(X).\label{q1}
	\end{equation}
	Substituting $\stks{\varkappa_{\gamma^{-1}},L}$ from \eqref{q1} in \eqref{V2} yields
	\begin{equation}
	\begin{aligned}
	V_\varepsilon[\varphi] = \frac12 \iint_{\singfund_{\varepsilon}} \frac{dx\wedge dy}{r^2}  +\frac{i}{8} \stks{\theta_{\gamma^{-1}}[\varphi],L}+ \pi \chi(X) \log\varepsilon,
	\end{aligned}
	\label{Veps}
	\end{equation}
	and putting equations \eqref{Aeps} and \eqref{Veps} together gives
	\begin{equation}
	\begin{aligned}
	V_\varepsilon[\varphi] - \frac12 A_\varepsilon[\varphi]-\pi \chi(X) \log \varepsilon = -\frac{1}{4} \Big(\frac{i}{2}
	\stks{\check{\omega} [\varphi], \singfund_{\varepsilon}} - \frac{i}{2} \stks{\theta_{\gamma^{-1}}[\varphi],L}
	\Big).
	\end{aligned}
	\label{tempeq1}
	\end{equation}
	Next, we will demonstrate that the terms in parentheses on the right-hand side of \eqref{tempeq1} correspond to the function $\Gpotential_{\boldsymbol{m}}$ (up to constant area term) described in \cite{Taghavi2024classical}. This function is described by Schottky uniformization of an orbifold Riemann surface with both cusp and conical singularities, which itself serves as the conformal boundary of the three-dimensional manifold $M$.  
	
	Based on the detailed explanation of Schottky uniformization of orbifold Riemann surfaces provided in Appendix \ref{Schotkkyreview}, the regularized Liouville action in the presence of singularities is (see \cite{zograf1988liouville,Takhtajan:2001uj,Taghavi2024classical})
	\begin{equation}
	\begin{split}
	&S_{\boldsymbol{m}}[\varphi] = S_{\boldsymbol{m}}(\SchottkyFund;z_1,\dots,z_{n}) 
	= S_{\singfund_{\text{reg}}}[\varphi] + \frac{\sqrt{-1}}{2} \sum_{k=2}^{g} \oint_{\mathcal{C}_k} \theta_{\gamma_k^{-1}}[\varphi],
	\end{split}\label{regularizeLiouvilleaction}
	\end{equation}
	with 
	\begin{multline}
	S_{\singfund_{\text{reg}}}[\varphi] = \lim_{\delta \to 0} \left(\iint_{\singfund_{\delta}} \Big(
	|\partial_z \varphi|^2 + e^\varphi
	\Big) d^2z + \underbrace{\frac{\sqrt{-1}}{2} \sum_{j=1}^{n_e} (1-\frac{1}{m_j}) \oint_{c_j^{\delta}} \varphi \Big(
		\frac{d\bar{z}}{\bar{z} - \bar{z}_j} - \frac{dz}{z-z_j}\Big) }_{\text{
			Call it $T_1$}}\right. \\ 
	\left. -2\pi \sum_{j=1}^{n_e} (1-\frac{1}{m_j})^2 \log \delta + 2\pi n_p \left(
	\log\delta + 2\log \left|\log \delta\right|
	\right)\right),
	\label{Sreg}
	\end{multline}
	and the generalized Liouville action is given by (see Section~4  in \cite{Taghavi2024classical} and Appendix \ref{Schotkkyreview})
	\begin{equation}
	\begin{aligned}
	\mathscr{S}_{\boldsymbol{m}}[\varphi] = S_{\boldsymbol{m}}[\varphi] - \pi \log H = S_{\boldsymbol{m}}[\varphi] - \overbrace{\pi\sum_{j=1}^{n} \log \mathbf{h}_j^{m_j h_j}}^{\text{Call it $T_2$}},
	\end{aligned}
	\label{renvol}
	\end{equation}
	where $\Lponetial_j$s are Hermitian metrics for tautological line bundles $\linebundle_j$ over $\schottky_{g,n}(\boldsymbol{m})$, and $h_j=1-1/m_j^2$ is conformal weight corresponding to the order of a marked point $m_j$. In the following, the key idea is to absorb the $T_1$ and $T_2$ terms, respectively, in \eqref{Sreg} and \eqref{renvol} into the radius of the regularized circles around the singularities. Based on  the asymptotic behaviour of the field $\varphi$ near the  singularities (see Appendix.\ref{asymapp}), we have
	\begin{equation}
	T_1  = 4\pi \sum_{j=1}^{n_e} \left(   (\frac{1}{m_j}-\frac{1}{m_j^2}) \log \left|J_1^{(j)}\right|+   (1-\frac{1}{m_j})^2 \log \delta \right),
	\label{bdterm}
	\end{equation}
	and
	\begin{equation}
	T_2 \overset{\eqref{loghj}}{=} 2\pi \sum_{j=1}^{n_e} (1-\frac{1}{m_j^2}) \log\left|J_{1}^{(j)}\right| +2\pi \sum_{j=n_e+1}^{n} \log\left|J_{1}^{(j)}\right|.
	\end{equation}
	Accordingly the $\mathscr{S}_{\boldsymbol{m}} $ in \eqref{renvol} becomes
	\begin{multline}
	\mathscr{S}_{\boldsymbol{m}}[\varphi]  = \lim_{\delta \to 0} \left(\iint_{\singfund_{\delta}} \Big(
	|\partial_z \varphi|^2 + e^\varphi
	\Big) d^2z - 2\pi \sum_{j=1}^{ne} (1-\frac{1}{m_j})^2 \log\left|J_{1}^{(j)}\right| -2\pi \sum_{j=n_e+1}^{n} \log\left|J_{1}^{(j)}\right|\right. \\ \left. + \frac{\sqrt{-1}}{2} \sum_{k=2}^{g} \oint_{\mathcal{C}_k} \theta_{\gamma_k^{-1}}[\varphi]+2\pi \sum_{j=1}^{n_e} (1-\frac{1}{m_j})^2 \log \delta + 2\pi n_p \left(
	\log\delta + 2\log \left|\log \delta\right|
	\right)\right).
	\label{Sreg2}
	\end{multline}
	Moreover, for each puncture and conical singularity one can see that
	\begin{equation}
	\frac{\sqrt{-1}}{2} \iint_{\delta}^{\left|J_{1}^{(j)}\right|\delta} \partial_z \varphi \partial_{\bar{z} }\varphi ~dz\wedge d\bar{z} = 2\pi \log\left|J_{1}^{(j)}\right|,\hspace{.2cm}(\text{cusp})
	\label{temp1}
	\end{equation}
	and
	\begin{equation}
	\frac{\sqrt{-1}}{2} \iint_{\delta}^{\left|J_{1}^{(j)}\right|\delta} \partial_z \varphi \partial_{\bar{z}}  \varphi ~dz \wedge d\bar{z} =  2\pi (1-\frac{1}{m_j})^2 \log\left|J_{1}^{(j)}\right|.\hspace{.2cm}(\text{conical singularity})
	\label{temp2}
	\end{equation}
	By utilizing \eqref{temp1} and \eqref{temp2} and noting that
	\begin{equation} \iint_{\delta}^{\infty}  -\iint_{\delta}^{\left|J_{1}^{(j)}\right|\delta}=\iint_{\left|J_{1}^{(j)}\right|\delta}^{\infty},\nonumber
	\end{equation}
	the expression \eqref{Sreg2} for $\mathscr{S}_{\boldsymbol{m}}$ simplifies to
	\begin{multline}
	\mathscr{S}_{\boldsymbol{m}}[\varphi]  =\lim_{\delta \to 0} \left( \iint_{\singfund_{\tilde\delta}} \Big(
	|\partial_z \varphi|^2 + e^\varphi
	\Big) d^2z + \frac{\sqrt{-1}}{2} \sum_{k=2}^{g} \oint_{\mathcal{C}_k} \theta_{\gamma_k^{-1}}[\varphi]\right. \\ \left. +2\pi \sum_{j=1}^{n_e} \left(1-\frac{1}{m_j}\right)^2 \log \delta + 2\pi n_p \left(
	\log\delta + 2\log \left|\log \delta\right|
	\right)\right),
	\label{Sreg3}
	\end{multline}
	where
	\begin{equation}
	\singfund_{\tilde\delta} = \singfund\backslash \bigcup_{j=1}^n \Big{\{} z\Big|~|z-z_j| < \left|J_1^{(j)}\right| \delta\Big{\}}.\label{q2}
	\end{equation}
	It's important to highlight that in this language, the divergence from the kinetic term is exactly minus of counterterms in the second line of \eqref{Sreg3}. Based on equation \eqref{Repsilon} and noting that $\varepsilon \rightarrow \delta$ as $r \rightarrow 0$, we can deduce from \eqref{tempeq1} and \eqref{q2} in this limit that
	\begin{equation}
	\begin{aligned}
	V_\varepsilon[\varphi]-\frac12 A_\varepsilon[\varphi] - \pi \chi(X) \log \varepsilon = -\frac{1}{4}~ \lim_{\delta \to 0}\Big(
	\iint_{\singfund_{\tilde\delta}} |\partial_z\varphi|^2 ~d^2z 
	+ \frac{\sqrt{-1}}{2} \sum_{k=2}^{g} \oint_{\mathcal{C}_k} \theta_{\gamma^{-1}_k}[\varphi]
	\Big).
	\end{aligned}
	\label{renvolrefined}
	\end{equation}
	
	Finally, it should be noted that the kinetic term on the right-hand side of \eqref{renvolrefined} is divergent, and the appropriate counterterms from the second line of \eqref{Sreg3} need to be added. Consequently, the renormalized volume can be defined as follows:
	\begin{equation}
	V_{\text{ren}}\hspace{-1mm}=\hspace{-.5mm}\lim_{\varepsilon \to 0}\hspace{-1mm}\left(\hspace{-1mm} V_{\varepsilon}[\varphi]\hspace{-.5mm} -\hspace{-.5mm} \frac{1}{2} A_{\varepsilon}[\varphi]\hspace{-.5mm} -\hspace{-.5mm} \pi \chi(X) \log \varepsilon\hspace{-.5mm} - \hspace{-.5mm}\frac{\pi}{2} \sum_{j=1}^{n_e} (1 - \frac{1}{m_j})^2\hspace{.5mm} \log \varepsilon\hspace{-.5mm} -\hspace{-.5mm} \frac{\pi}{2} n_p (\log \varepsilon + 2 \log \left|\log \varepsilon\right|)\hspace{-1mm} \right),\label{Vren}
	\end{equation}
	which, when compared with the generalized Liouville action \eqref{Sreg3}, implies that\footnote{The term "area" here refers to the contribution from the expression $\iint_{\singfund_{\tilde\delta}}  e^\varphi  d^2z$, which is equal to $-2\pi\chi(X)$.}
	\begin{equation}
	V_{\text{ren}} = -\frac{1}{4} \check{\mathscr{S}}_{\boldsymbol{m}}[\varphi]=-\frac{1}{4} \Big(\mathscr{S}_{\boldsymbol{m}}[\varphi]-\text{area term}\Big).\label{Theorem1}
	\end{equation}
	This indicates that the generalized Liouville action  $\Gpotential_{\boldsymbol{m}}$  satisfies the holography principle, meaning it serves as the renormalized hyperbolic volume of a Schottky 3-manifold with lines of conical singularities and punctures. Moreover, we obtain an alternative proof that the  $\mathscr{S}_{\boldsymbol{m}}$ is independent of the choice of a fundamental domain $\singfund$ 
	of Schottky group $\Sigma$ 
	in $\singrigon$, as long as $\singfund$  is the boundary of a fundamental region of $\Sigma$ in 
	$R_{\varepsilon}\cup \singrigon $. Furthermore, as demonstrated in \cite{Taghavi2024classical}, the generalized Liouville action \eqref{Sreg3} serves as a Kähler potential for a specific combination of the Weil–Petersson and Takhtajan–Zograf metrics,\footnote{See also the second equation in \eqref{derivativesL}.}  which appear in the local index theorem for orbifold Riemann surfaces \cite{takhtajan2019local}. The result \eqref{Theorem1} indicates that $V_{\text{ren}}$ plays the same role.
	\section{Renormalized Volume and Polyakov Anomaly}\label{Section3}
	In this section, we show that, under conformal transformations, the variations in the function $\Gpotential_{\boldsymbol{m}}$ and its dual renormalized hyperbolic volume are equivalent to the Polyakov anomaly for orbifold Riemann surfaces, as previously derived in \cite{kalvin2021polyakov,aldana2020polyakov}. The method we used to establish this equivalence may provide an alternative approach to derive the renormalized Polyakov anomaly for Riemann surfaces with punctures (cusps), as described in \cite{albin2013ricci}, regardless of whether they have conical singularities.
	
	Here, we have assumed that all singular points except the last one, i.e., $z_n$, are located at a finite distance within the interior of the fundamental domain  $\SchottkyFund$. For sufficiently small $\delta >0$, define
	\begin{equation}
	\singfund_{\delta} = \singfund \big\backslash  \bigcup_{i=1}^{n-1}\left\{z \, \Big| \, |z-z_i|<\delta\right\} \cup \left\{z \, \Big| \, |z|>\delta^{-1}\right\}.
	\end{equation}
	Accordingly, the $S_{\singfund_{\text{reg}}}[\varphi]$ in \eqref{Sreg} changes to
	\begin{multline}\label{Sb}
	S_{\singfund_{\text{reg}}}[\varphi] = 
	\lim_{\delta \to 0^{+}} \left(\iint_{\singfund_{\delta}}(|\partial_z \varphi|^2 + e^{\varphi}) \dd^2{z}  + \frac{\sqrt{-1}}{2} \sum_{j=1}^{n_e} \left(1-\frac{1}{m_j}\right) \oint_{c_{j}^{\delta}} \varphi \left(\frac{\dd{\bar{z}}}{\bar{z}-\bar{z}_{j}} - \frac{\dd{z}}{z - z_{j}}\right)\right. \\ 
	\left. - 2\pi \sum_{j=1}^{n_e} \left(1-\frac{1}{m_j}\right)^2 \log\delta    + 2 \pi n_p \log\delta + 4 \pi (n_p-2) \log|\log\delta|\big) \right).
	\end{multline}
	By including the appropriate boundary term for a well-defined variational principle, even in the case of punctures, $S_{\singfund_{\text{reg}}}$ can be expressed as follows (see Remark 4.2 in \cite{Taghavi2024classical} for further details).
	\begin{multline}\label{FullNoGenusAction}
	S_{\singfund_{\text{reg}}}[\varphi] = \\
	\lim_{\delta \to 0} \left(\iint_{\singfund_{\delta}}(|\partial_z \varphi|^2 + e^{\varphi}) \dd^2{z}  + \underbrace{\frac{\sqrt{-1}}{2} \sum_{i=1}^{n} \left(1-\frac{1}{m_i}\right) \oint_{c_{i}^{\delta}} \varphi \left(\frac{\dd{\bar{z}}}{\bar{z}-\bar{z}_{i}} - \frac{\dd{z}}{z - z_{i}}\right)}_{B_1[\varphi]}\right. \\ 
	+\underbrace{\frac{\sqrt{-1}}{2} \sum_{j=n_e+1}^{n-1}  \oint_{c_{j}^{\delta}} \varphi \left(\frac{\dd{\bar{z}}}{(\bar{z}-\bar{z}_{j})\log|z-z_j|} - \frac{\dd{z}}{(z-z_j)\log|z-z_j|} \right)}_{B_2[\varphi]}  \\
	\left.\underbrace{-\frac{\sqrt{-1}}{2} \oint_{c_{n}^{\delta}} \varphi \left(\frac{\dd{\bar{z}}}{\bar{z}} - \frac{\dd{z}}{z}\right) -\frac{\sqrt{-1}}{2} \oint_{c_{n}^{\delta}} \varphi \left(\frac{\dd{\bar{z}}}{\bar{z}\log|z|} - \frac{\dd{z}}{z\log|z|}\right)}_{B_3[\varphi]} -2\pi \sum_{i=1}^{n} \left(1-\frac{1}{m_i}\right)^2 \log\delta 
	\right),
	\end{multline}
	where $m_i = \infty$ for $i=n_e+1,\dots,n$. Now, the regularized action becomes 
	\begin{equation}
	\begin{split}
	&S_{\boldsymbol{m}}[\varphi] = S_{\boldsymbol{m}}(\SchottkyFund;z_1,\dots,z_{n}) 
	= S_{\singfund_{\text{reg}}}[\varphi] + \frac{\sqrt{-1}}{2} \sum_{k=2}^{g} \oint_{\mathcal{C}_k} \theta_{\gamma_k^{-1}}[\varphi],
	\end{split}\label{regularizeLiouvilleaction2}
	\end{equation}
	where $S_{\singfund_{\text{reg}}}[\varphi]$ and the 1-form $\theta_{L_k^{-1}}(\varphi)$ are given in \eqref{FullNoGenusAction} and \eqref{chigm1}, respectively. 
	
	Now, under the conformal transformation $e^{\varphi}\dd z \dd{\bar{z}} \rightarrow e^{\varphi+\sigma}\dd z \dd{\bar{z}}$, the variation of Kinetic term $\check\omega[\varphi]$ in \eqref{omegacheck} is given by
	\begin{equation}\label{varcheckomega}
	\check\omega[\varphi +\sigma]  -\check\omega[\varphi] = \Big(\partial_z\sigma\partial_{\bar{z}} \sigma +  K[\varphi] ~e^{\varphi} \sigma\Big) \dd z\wedge \dd \bar{z} +\dd \tilde{\theta},
	\end{equation}
	where 
	\begin{equation}\label{thetat}
	\tilde{\theta} = \sigma \left(\partial_{\bar{z}}\varphi ~\dd \bar{z}-\partial_{z}\varphi~ \dd z\right) \overset{\eqref{chi1}}{=}\sigma \varkappa_{1}.
	\end{equation}
	Observing that $\delta(\varkappa_{1})_{\gamma^{-1}} = \varkappa_{1}.\gamma-\varkappa_1$, we get
	\begin{equation}\label{deltathetat}
	\delta\tilde{\theta}_{\gamma^{-1}} = \delta (\sigma \varkappa_{1}) =\sigma \delta\varkappa_{1}= \sigma \left(\frac{\gamma^{\prime\prime}}{\gamma^{\prime}} \dd w -\frac{\overline{\gamma^{\prime\prime}}}{\overline{\gamma^{\prime}}} \dd\bar{w}\right) \overset{\eqref{chigm1}}{=} \sigma \varkappa_{\gamma^{-1}} 
	\end{equation}
	According to the equation \eqref{chigm1}, we have
	\begin{equation}\label{vartheta}
	\theta_{\gamma_k^{-1}}[\varphi+\sigma]-\theta_{\gamma_k^{-1}}[\varphi] = \sigma \left(\frac{\gamma^{\prime\prime}}{\gamma^{\prime}} \dd w -\frac{\overline{\gamma^{\prime\prime}}}{\overline{\gamma^{\prime}}} \dd\bar{w}\right) \overset{\eqref{chigm1}}{=} \sigma \varkappa_{\gamma^{-1}} = \delta\tilde{\theta}_{\gamma^{-1}},
	\end{equation}
	where the last equality follows from equation \eqref{deltathetat}. Thus,
	\begin{multline}
	\stks{\check{\omega}[\varphi+\sigma]-\check{\omega}[\varphi],\singfund_{\delta}}-\stks{\theta[\varphi+\sigma]-\theta[\sigma],L} = \\ \iint_{\singfund_{\delta}}\Big(\partial_z \sigma\partial_{\bar{z}}\sigma +  K[\varphi] ~e^{\varphi} \sigma\Big) \dd z\wedge \dd \bar{z} +\sum_{j=1}^{n}\stks{\tilde{\theta}, c_{j}^{\delta}},\label{q3}
	\end{multline}
	where we have used the equations \eqref{varcheckomega},\eqref{vartheta}, and the fact that (see also \eqref{parD})
	\begin{equation}
	\stks{\dd \tilde\theta,\singfund_{\delta}} =\stks{\tilde\theta,\partial' \singfund_{\delta}} = \stks{\tilde\theta, \partial'' L+ \sum_{j=1}^{n}c_{j}^{\delta}} = \stks{\delta\tilde\theta,L}+\sum_{j=1}^{n}\stks{\tilde\theta,c_{j}^{\delta}}.
	\end{equation}
	Using the equation \eqref{thetat} along with the asymptotic form of the field $\varphi$ near the singularities (see equation \eqref{conicalasymp}), it is straightforward to observe that
	\begin{equation}
	\sum_{j=1}^{n}\stks{\tilde\theta,c_{j}^{\delta}} = - \Big(B_{1}[\varphi+\sigma]-B_{1}[\varphi]\Big)- \Big(B_{2}[\varphi+\sigma]-B_{2}[\varphi]\Big)- \Big(B_{3}[\varphi+\sigma]-B_{3}[\varphi]\Big).\label{q4}
	\end{equation}
	Therefore, according to \eqref{q3} with \eqref{q4} and for $\check S_{\boldsymbol{m}}[\varphi]=S_{\boldsymbol{m}}[\varphi]-\text{area term}$ (see equation \eqref{regularizeLiouvilleaction2}), we have
	\begin{equation}
	\check S_{\boldsymbol{m}}[\varphi+\sigma]- \check S_{\boldsymbol{m}}[\varphi] = \iint_{\singfund_{\delta}}\Big(\partial_z \sigma\partial_{\bar{z}}\sigma +  K[\varphi] ~e^{\varphi} \sigma\Big) \dd^{2}z.\label{q5}
	\end{equation}
	Moreover, in this case, the function $H$ in equation \eqref{renvol} is modified to
	\begin{equation}
	\mathsf{H}[\varphi] = \Lponetial_{1}^{m_1h_1}\dotsm\Lponetial_{n_e}^{m_{n_e} h_{n_e}}\Lponetial_{n_e+1}\dotsm\Lponetial_{n-1}\Lponetial^{-1}_{n},\label{H2}
	\end{equation}
	where (see Remark 3.10 in \cite{Taghavi2024classical})
	\begin{equation}\label{loghj}
	\left\{
	\begin{split}
	& \log \Lponetial_{j} = -2 \log m_j + 2 \log 2 - \lim_{z \to z_{j}}\left(\varphi(z) + \left(1-\frac{1}{m_j}\right) \log|z-z_{j}|^2\right), \hspace{.2cm} j=1,\dots,n_e,\\ \\
	&\log \Lponetial_{j} = \lim_{z \to z_{i}} \left(\log|z-z_{j}|^2 - \frac{2 e^{-\frac{\varphi(z)}{2}}}{|z-z_{j}|}\right),\hspace{2cm}  j=n_e+1,\dots,n-1,\\
	& \log \Lponetial_{n} = \lim_{z \to \infty} \left(\log|z|^2 - \frac{2 e^{-\frac{\varphi(z)}{2}}}{|z|}\right).
	\end{split}
	\right.
	\end{equation}
	with $\Lponetial_j = \left| J^{(j)}_{1} \right|^{\frac{2}{m_j}}$ for $j=1,\dots,n_e$,  $\Lponetial_{j} =\left| J^{(j)}_{1} \right|^{2}$ for $j=n_e+1,\dots,n-1$ and $\Lponetial_{n}= \left| J^{(n)}_{-1} \right|^2$ for $j=n$.
	
	Let us first consider the case where there are no punctures (i.e., $n_p=0$). It is easy to see that
	\begin{equation}
	\pi \log \mathsf{H}[\varphi+\sigma]-\pi \log \mathsf{H}[\varphi] = -\pi \sum_{j=1}^{n_e} m_j h_j~ \sigma(z_j).\label{q6}
	\end{equation}
	Thus, according to \eqref{q5} with \eqref{q6} and for $\check{\mathscr{S}}_{\boldsymbol{m}}$ in \eqref{Theorem1} (see also equation \eqref{renvol}), we get
	\begin{equation}
	\check{\mathscr{S}}_{\boldsymbol{m}}[\varphi+\sigma]-\check{\mathscr{S}}_{\boldsymbol{m}}[\varphi] = \iint_{\singfund_{\delta}}\Big(\partial_z \sigma\partial_{\bar{z}}\sigma +  K[\varphi] ~e^{\varphi} \sigma\Big) \dd^{2}z +\pi \sum_{j=1}^{n_e} m_j h_j~ \sigma(z_j).\label{q7}
	\end{equation}
	According to Corollary 1.3.1 in \cite{kalvin2021polyakov} or Theorem 1.3 in \cite{aldana2020polyakov} for the metric $\tilde{g}= e^{2\varphi_{0}}g$, the Polyakov anomaly formula is given by\footnote{In this context, $\Delta_{g}$ refers to the Friedrichs extension of the Laplacian associated with the Riemannian metric $g$.}
	\begin{multline}
	\log\frac{\text{det}(\Delta_{\tilde{g}})}{A_{\tilde{g}}}- \log\frac{\text{det}(\Delta_{g})}{A_{g}}\\
	=-\frac{1}{12\pi}\Big(\iint_{M}\left(|\nabla_{g}\varphi_0|^2 +  R_{g}\hspace{.5mm}\varphi_0\right) dA_g+ \sum_{j=1}^{n_
		e}\varphi_0(z_j)\frac{(2\pi)^2-\gamma_{j}^2}{\gamma_j}\Big).\label{p1}
	\end{multline}
	By choosing the metric $g= e^{2\tilde{\varphi}}\eta$ and $\eta= \dd z \dd{\bar{z}}$, and noting that
	\begin{equation}
	\begin{split}
	&\Delta_{\eta}\tilde{\varphi} = 4\partial_{z}\partial_{\bar{z}}\tilde{\varphi},\\&
	R_{g} = -2e^{-2\tilde{\varphi}}\Delta_{\eta}\tilde{\varphi} = -8e^{-2\tilde{\varphi}}\partial_{z}\partial_{\bar{z}}\tilde{\varphi},\\&
	|\nabla_{g}\varphi_0|^2 = (\partial_{\mu}\varphi_0)(\partial^{\mu}\varphi_0) =  4 e^{-2\tilde{\varphi}}\partial_{z}\varphi_0\partial_{\bar{z}}\varphi_0,\\&
	dA_g = e^{2\tilde{\varphi}}\dd z\dd\bar{z},
	\end{split}
	\end{equation}
	the Polyakov anomaly \eqref{p1} is simplified to
	\begin{multline}
	\log\frac{\text{det}(\Delta_{\tilde{g}})}{A_{\tilde{g}}}- \log\frac{\text{det}(\Delta_{g})}{A_{g}}\\
	=-\frac{1}{12\pi}\left(\iint_{M}\Big(4 \partial_{z}\varphi_0\partial_{\bar{z}}\varphi_0   -8\partial_{z}\partial_{\bar{z}}\tilde{\varphi} ~\varphi_0\Big) d^{2}z+ \sum_{j=1}^{n_
		e}\varphi_0(z_j)\frac{(2\pi)^2-\gamma_{j}^2}{\gamma_j}\right).\label{p2}
	\end{multline}
	By defining 
	\begin{equation}
	P[\tilde\varphi] =\log \frac{\text{det}(\Delta_{g})}{A_{g}},
	\end{equation} 
	and noting that in our notation $\varphi_0 = \sigma/2, \tilde{\varphi}=\varphi/2$, $\gamma_{j}= 2\pi/m_j$ and $M\equiv \singfund_{\delta}$, the equation \eqref{p2} can be rewritten as follows
	\begin{equation}
	P[\varphi+\sigma]-P[\varphi]\\
	=-\frac{1}{12\pi}\left(\iint_{\singfund_{\delta}}\Big( \partial_{z}\sigma\partial_{\bar{z}}\sigma   +K[\varphi] e^{\varphi}\sigma\Big) d^{2}z+ \pi\sum_{j=1}^{n_
		e}m_j h_j\sigma(z_j)\right).\label{p3}
	\end{equation}
	Comparing it with equation \eqref{q7} we get
	\begin{equation}
	P[\varphi+\sigma]+\frac{1}{12\pi}\check{\mathscr{S}}_{\boldsymbol{m}}[\varphi+\sigma]=P[\varphi]  +\frac{1}{12\pi}\check{\mathscr{S}}_{\boldsymbol{m}}[\varphi]\label{Theorem2}.
	\end{equation}
	Additionally, this comparison with equation \eqref{Theorem1} implies
	\begin{equation}
	V_{\text{ren}}[\varphi+\sigma]-3\pi \hspace{.5mm}P[\varphi+\sigma]=V_{\text{ren}}[\varphi]-3\pi\hspace{.5mm}P[\varphi].\label{Theorem2p}
	\end{equation}
	Thus, the renormalized volume $V_{\text{ren}}$ follows a Polyakov formula and acts as an analog to the determinant of the Laplacian, serving as an action on the conformal class of boundary metrics, with critical points occurring at constant curvature metrics. It is important to recall that for the case of a compact Riemann surface with genus greater than one, Takhtajan and Teo in \cite{Takhtajan_2003} derived the following equation:\footnote{For an extension to higher dimensions see \cite{guillarmou2018renormalized}.}
	\begin{equation}
	P[\varphi+\sigma]+\frac{1}{12\pi}\check S[{\varphi+\sigma}]=P[\varphi]  +\frac{1}{12\pi}\check S[\varphi]\label{TakhtajanTeo},
	\end{equation}
	where
	\begin{equation}
	\check S[\varphi] = \iint_{\SchottkyFund} |\partial_z \varphi|^2 \dd^2{z}+ \frac{\sqrt{-1}}{2} \sum_{k=2}^{g} \oint_{\mathcal{C}_k} \theta_{\gamma_k^{-1}}[\varphi],
	\end{equation}
	and the 1-form $\theta_{L_k^{-1}}(\varphi)$ is given in \eqref{chigm1}. By comparing equation \eqref{TakhtajanTeo} with our result in equation \eqref{Theorem2}, it becomes clear that to extend the findings of Takhtajan and Teo for compact Riemann surfaces to orbifold Riemann surfaces, one must replace $\check S[\varphi]$, the Liouville action without the area term on Schottky space $\schottky_{g}$, with the corresponding function on $\schottky_{g,n}(\boldsymbol{m})$. This function is $\check{\mathscr{S}}_{\boldsymbol{m}}[\varphi]$, the generalized Liouville action $\Gpotential_{\boldsymbol{m}}[\varphi]$, introduced in \cite{Taghavi2024classical}, also without the area term.
	
	Now, let us proceed to extend the analysis to include the case with punctures. According to equation \eqref{loghj}, the contribution of punctures to the function $H$ in \eqref{H2} is given by
	\begin{equation}
	2\pi\sum_{j=n_e+1}^{n-1}\lim_{z\rightarrow z_j}\left(1-e^{-\sigma/2}\right)\log \left|\frac{z-z_j}{J_{1}^{(j)}}\right|-2\pi\lim_{z\rightarrow\infty}\left(1-e^{-\sigma/2}\right)\log\left|\frac{z}{J_{-1}^{(n)}}\right|,
	\end{equation}
	which implies that for $|z-z_j|=\epsilon\rightarrow 0$,
	\begin{equation}
	\pi \log \mathsf{H}[\varphi+\sigma]-\pi \log \mathsf{H}[\varphi] = -\pi \sum_{j=1}^{n_e} m_j h_j~ \sigma(z_j)+2\pi\sum_{j=n_e+1}^{n}\left(1-e^{-\sigma(z_j)/2}\right)\log \epsilon.\label{H3}
	\end{equation}
	In the equation above, we applied a change of variables so that the puncture originally at infinity is now located at $z_n=0$.\footnote{In the case where all punctures are at finite distances, the last term in \eqref{H3} remains unchanged.}  Assuming that equation \eqref{Theorem2} remains valid even in the presence of punctures, we define the renormalized Polyakov anomaly according to \eqref{H3} as follows\footnote{Where $\epsilon= |z-z_j|$.}
	\begin{equation}
	\Big(P[\varphi+\sigma]-P[\varphi]\Big)_{\text{ren}}\hspace{-.2cm}=\Big(P[\varphi+\sigma]-P[\varphi]\Big)
	-\frac{1}{6}\lim_{\epsilon\rightarrow 0}\sum_{j=n_e+1}^{n}\left(1-e^{-\sigma(z_j)/2}\right)\log \epsilon,
	\end{equation}
	which leads to:
	\begin{equation}
	\Big(P[\varphi+\sigma]-P[\varphi]\Big)_{\text{ren}}\hspace{-.2cm}= -\frac{1}{12\pi}\left(\iint_{\singfund_{\delta}}\Big( \partial_{z}\sigma\partial_{\bar{z}}\sigma   -K[\varphi] e^{\varphi}\sigma\Big) d^{2}z+ \pi\sum_{j=1}^{n_e}m_j h_j\hspace{.5mm}\sigma(z_j)\right).\label{p4}
	\end{equation}
	Unlike conical singularities, where additional finite terms appear in the  Polyakov anomaly compared to the compact case, punctures (cusps) do not introduce any new finite terms. This finding aligns with Theorem 2.9, particularly equation (2.17), from the study by Albin, Aldana, and Rochon in \cite{albin2013ricci}.\footnote{See also \cite{aldana2013asymptotics}.} More importantly, this suggests that transitioning from a conical singularity to a cusp (as the cone angle approaches zero) does not simply yield the same renormalized determinant as the cusp case. Notably, in the compact scenario, when we pinch a geodesic to create cusp ends, a result from Wolpert \cite{wolpert1987asymptotics} indicates that the determinant diverges in this limit, tending to infinity.\footnote{We appreciate Frederic Rochon for a discussion regarding this matter.}
	\section{Conclusion}\label{Conclusion}
	In this paper, we established the holography principle, demonstrating a precise connection between the function (generalized Liouville action) $\Gpotential_{\boldsymbol{m}}$, introduced in \cite{Taghavi2024classical} for orbifold Riemann surfaces and the renormalized hyperbolic volume of the associated Schottky 3-orbifolds. Let $X \cong \singrigon \slash \Sigma = [\UHP \slash \Gamma]$ and $\Gamma$ be respectively an orbifold Riemann surface and Fuchsian group of signature  $(g;m_1,\dots,m_{n_e};n_p)$ and let $J : \UHP \to \singrigon$ be the corresponding orbifold covering map. As it is explained in Section 5.2  of \cite{Taghavi2024classical}, the automorphic form $\text{Sch}\left(J^{-1};z\right)$ of weight four for the Schottky group $\Sigma$ can be projected to subspace of its meromorphic quadratic differential  $\Hilbert^{2,0}(\singrigon,\Sigma)$ with basis element $P_i$ as follows,
	\begin{equation*}\label{Rquaddiffexpansion}
	\mathsf{R}(z) = \sum_{i=1}^{3g-3+n} \mathrm{b}_i\hspace{.5mm} P_i(z) = \sum_{i=1}^{3g-3+n} \left( \text{Sch}\left(J^{-1}\right), M_i \right) P_i(z),
	\end{equation*}
	where $M_i$ is basis element for harmonic differentials in $\Hilbert^{-1,1}(\singrigon,\Sigma)$. The $\mathsf{R}(z)$ coincides with a $(1,0)$-form $\mathscr{Q}$ on the Schottky pace $\schottky_{g,n}(\boldsymbol{m})$,\hspace{1mm}
	\begin{equation*}\label{D}
	\begin{split}
	\mathscr{Q}= \sum_{i=1}^{3g-3+n}\mathrm{b}_i \hspace{.5mm}\dd z_i & = \mathrm{b}_1\hspace{.5mm}\dd\lambda_1+\dots +\mathrm{b}_g \hspace{.5mm}\dd\lambda_g+\mathrm{b}_{g+1}\hspace{.5mm}\dd a_3+\dots+\mathrm{b}_{2g-3}\hspace{.5mm}\dd a_g\\
	&+\mathrm{b}_{2g-2}\hspace{.5mm}\dd b_2+\dots+\mathrm{b}_{3g-3}\hspace{.5mm}\dd b_g+\mathrm{b}_{3g-2}\hspace{.5mm}\dd z_1+\dots+\mathrm{b}_{3g-3+n}\hspace{.5mm}\dd z_n.
	\end{split}
	\end{equation*}
	According to Theorem 1 and Theorem 2 in \cite{Taghavi2024classical}, the function $\Gpotential_{\boldsymbol{m}}$ satisfies the following equations:
	\begin{equation}
	\begin{split}
	&\partial \Gpotential_{\boldsymbol{m}} = 2 \mathscr{Q},\\&\bar{\partial} \partial \Gpotential_{\boldsymbol{m}} = 2\sqrt{-1}\left(\omega_{\text{WP}}-\frac{4\pi^2}{3} \omega^{\text{\text{cusp}}}_{\text{TZ}}- \frac{\pi}{2} \sum_{j=1}^{n_e} m_j h_j \hspace{.5mm}\omega^{\text{\text{ell}}}_{\text{TZ},j}\right),\label{derivativesL}
	\end{split}
	\end{equation}
	where $\omega_{\text{WP}}$ and $(\omega^{\text{\text{cusp}}}_{\text{TZ}},\omega^{\text{\text{ell}}}_{\text{TZ}})$ are Weil--Petersson and Takhtajan--Zograf metrics, respectively. Given that the  $\Gpotential_{\boldsymbol{m}}$ is linked to the renormalized volume, or equivalently, to the renormalized Einstein-Hilbert action with  particles, it is important to examine the significance of its first and second derivatives in equation \eqref{derivativesL}. These derivatives could correspond to key physical concepts in the bulk, such as black hole thermodynamics, intensive and extensive quantities, stability of black hole solutions, phase transitions between different gravitational configurations, and counting microscopic states of the black hole (which leads to its entropy). Additionally, they contribute to the broader aim of quantum gravity, which seeks to uncover the quantum structure of spacetime and the nature of black hole singularities.

	We have also shown that the relationship between changes in the logarithm of the determinant of the Laplacian and the Liouville action under conformal transformations, as observed by Takhtajan and Teo in \cite{Takhtajan_2003} for compact Riemann surfaces, can be extended to orbifold Riemann surfaces using the  $\Gpotential_{\boldsymbol{m}}$. As discussed in the introduction, these relationships are two sides of the same coin—namely, uniformization theory. However, they also reveal deeper insights. It is important to highlight that surfaces with different shapes or geometries can have identical Laplace spectra (yielding the same value for $-\log\det\Delta$), known as isospectral surfaces.\footnote{An example of isospectral surfaces with hyperbolic metrics can be found in the classical constructions by Buser, Conway, Doyle, and Semmler \cite{buser2010some}, who provided examples of isospectral but non-isometric surfaces. See also \cite{vigneras1980varietes,sunada1985riemannian,gordon1992one}.} A natural question that arises is how these surfaces are understood through the lens of the  $\Gpotential_{\boldsymbol{m}}$. To explore this, let us consider two representatives of the set $\Sigma.\{z_1,...,z_n\}$. By choosing different generators for the Schottky group, one can obtain distinct fundamental domains that, although topologically identical, describe different conformal structures on the same Riemann surface. These variations are linked to distinct points on the moduli (or Schottky) space. Now, consider how the function $\Gpotential_{\boldsymbol{m}}$ behaves in these cases. According to equation \eqref{Sreg3}, the effective fundamental domain for calculating  $\Gpotential_{\boldsymbol{m}}$ is given by
	\begin{equation}
	\singfund_{\tilde\delta} = \singfund\backslash \bigcup_{j=1}^n \Big{\{} z\Big|~|z-z_j| < \left|J_1^{(j)}\right| \delta\Big{\}}.\label{D1}\nonumber
	\end{equation}
	When a generator $L_k\in\Sigma$ acts on a singular point $z_i$
	within a specific fundamental domain $\singfund$, it maps the point to another singularity with the same ramification index but in a different domain. Although the fundamental domain $\singfund$ changes, the radius $\tilde\delta$ used in calculating the  $\Gpotential_{\boldsymbol{m}}$ remains invariant. This invariance can be seen easily by noting that under $z_j \rightarrow L_k(z_j)$ we have (see \cite{Taghavi2024classical} for more details)
	\begin{equation}
	\delta \rightarrow \delta/\left|L'_k(z_j)\right|,\hspace{1cm}\left|J_{1}^{(j)}\right| \rightarrow \left|J_1^{j}\right| \left|L'_k(z_j)\right|.\nonumber
	\end{equation}
	Since  $\Gpotential_{\boldsymbol{m}}$ is connected to the spectrum of the operator  $-\log\det\Delta$, this implies that the spectrum remains unchanged when different representatives of $\Sigma.\{z_1,...,z_n\}$ are chosen. In other words, the orbifold Riemann surfaces corresponding to different representatives are isospectral. Additionally, this indicates that these orbifold Riemann surfaces share the same lengths of closed geodesics since it has been demonstrated that in \cite{doyle2006isospectral}, and more broadly through the Selberg Trace Formula,  that two-dimensional hyperbolic surfaces, including both manifolds and orbifolds, that are isospectral have matching geodesics.
	
	It is also useful to provide a remark on the various logarithmic divergent terms in the definition of $V_{\text{ren}}$ in equation \eqref{Vren}. For a CFT with an energy-momentum tensor $T_{\mu\nu}$ and central charge $c$ on an orbifold Riemann surface, the conformal anomaly remains proportional to the Euler characteristic of the surface, expressed as:
	\begin{equation}
	\mathcal{A}=\int_{X} \left\langle T_{\mu}^{\mu}\right\rangle \dd A = -\frac{c}{24\pi}\hspace{1mm}.\hspace{1mm} 4\pi\chi(X).\label{CA1}
	\end{equation}
	For a CFT dual to the bulk considered in this paper, with $c = 3/2G_N$, this conformal anomaly \eqref{CA1} simplifies to:
	\begin{equation}
	\mathcal{A} = -\frac{1}{4G_N} \chi(X).\label{CA2}
	\end{equation}
	Moreover, the coefficient of the $\log\varepsilon$ term in on-shell Einstein-Hilbert action is (see also \cite{henningson1998holographic, krasnov2000holography}):
	\begin{equation}
	\frac{1}{4\pi G_{N}} V_{\varepsilon}\bigg{|}_{\log{\varepsilon}}=\frac{1}{4\pi G_{N}}\hspace{1mm}.\hspace{1mm}\pi\chi(X)=\frac{1}{4G_N}\chi(X), 
	\end{equation}
	where equation \eqref{Veps} is used. This indicates that the term proportional to the Euler characteristic $\chi(X)$ in $V_{\text{ren}}$, as seen in result \eqref{Vren}, corresponds to the conformal anomaly for orbifold Riemann surfaces. Additionally, since
	\begin{equation} 
	\chi(X)=(-1/2\pi) \iint_{\singfund_{\tilde\delta}}  e^\varphi  d^2z,\nonumber
	\end{equation}
	the conformal anomaly \eqref{CA2} can be connected to the following counterterm:
	\begin{equation}
	S_{\text{ct}}= \frac{1}{8\pi G_N} \iint_{\singfund_{\tilde\delta}}  e^\varphi  d^2z\hspace{1mm} \log\varepsilon.\label{Sct}
	\end{equation} 
	This is where the first logarithmic term, containing the Euler characteristic in the definition of $V_{\text{ren}}$, conceptually differs from the other two logarithmic divergent terms, which are field-independent. The counterterm \eqref{Sct} is field-dependent because the anomaly arises from changes in the underlying geometry or field configuration under a conformal transformation, which is captured by the field $\varphi$. In contrast, field-independent terms, which do not depend on the conformal factor of the metric, are constants and do not influence the specific way the theory changes under conformal transformations.

Furthermore, the renormalization procedure used in this manuscript involves the uniformization of surfaces at the conformal boundary. Therefore, the renormalized volume \eqref{Theorem1} is a Schottky invariant quantity. A superficial analogy exists between the mass of asymptotically Euclidean manifolds and $V_{\text{ren}}$. Similar to the positive mass conjecture, one might ask whether $V_{\text{ren}}$ is positive for all hyperbolic Schottky manifolds with lines of conical singularities or whether any bounds exist on it.
	
Lastly, consider deformations (i.e., changes in the conformal structure) of a smooth cutoff surface $f=\varepsilon$ within a three-dimensional manifold $M$. In the compact case, it was shown in \cite{rivin1999schlafli} that there exists a relationship between the variations $dV$, $dH$, and $dI$ of the volume bounded by $f$, mean curvature, and induced metric on  $f$. This relationship leads to the  Schl{\"a}fli formula for polyhedron. To recall, in three-dimensional space, the  Schl{\"a}fli formula describes how the volume 
	$V$ changes with respect to infinitesimal variations in the polyhedron’s dihedral angles $\theta_i$
	\begin{equation}
	d V=-\frac{1}{2}\sum_{k} l_k\hspace{1mm} d\theta_k, \label{Schlafli}
	\end{equation}
	where $l_k$ denotes the length of the
	$k$-th edge of the polyhedron, $\theta_k$ is the dihedral angle at that edge, and the summation is over all edges of the polyhedron. A Polyhedron with the identification of some faces can model a three-dimensional space containing lines of conical singularities. Consequently, by utilizing the connection between the renormalized volume $V_{\text{ren}}$
	and the function $\Gpotential_{\boldsymbol{m}}$ in \eqref{Theorem1}, along with the first equation in \eqref{derivativesL}, one can extend the results of \cite{rivin1999schlafli} to account for lines of conical singularities. This allows for deriving an extended version of the Schl{\"a}fli formula for polyhedra with edge singularities. In this manner, we gain a more geometric understanding of the variations in the renormalized volume and its connection to classical thermal entropy and quantum information theoretic quantities, such as holographic complexity.

	\noindent {\bf Acknowledgments}:

	The authors would like to thank Colin Guillarmou, Victor Kalvin and Frederic Rochon for inspiring discussions. B.T. would like to thank CERN theory department for hospitality where part of this work has been carried out. The research of A.N. and B.T. is supported by the Iranian National Science Foundation (INSF) under Grant No.~4001859.

	\appendix
	\section{Asymptotic form of $\varphi$ near punctures and conical points}\label{asymapp}
	For an orbifold with conical points and punctures labeled by $n=1,\dots n_e$, and $n=n_e+1,\dots, n= n_e+n_p$, and located at $z=z_i$ respectively, we have the following asymptotic expansions for Liouville field and its derivative (see Appendix.C in \cite{Taghavi2024classical} for more details). Notice that the last puncture point is fixed to be at $\infty$.
	
	\noindent For $i=1, \ldots, n_e$ and $j=n_e+1, \ldots, n-1$,
	\begin{equation}
	\begin{aligned}
	\varphi(z)= \begin{cases}-2\left(1-\frac{1}{m_i}\right) \log \left|z-z_i\right|+\log \frac{4\left|J_1^{(i)}\right|^{-\frac{2}{m_i}}}{m_i^2}+\smallO(1) & z \rightarrow z_i \\ -2 \log \left|z-z_j\right|-2 \log |\log | \frac{z-z_j}{J_1^{(j)}}||+\smallO(1) & z \rightarrow z_j, \\ -2 \log |z|-2 \log \log \left|\frac{z}{J_{-1}^{(n)}}\right|+\smallO\left(|z|^{-1}\right), & z \rightarrow \infty\end{cases}
	\end{aligned}
	\label{conicalasymp}
	\end{equation}
	For $i=1, \ldots, n_e$ and $j=n_e+1, \ldots, n-1$,
	\begin{equation}
	\begin{aligned}
	\partial_z \varphi(z)= \begin{cases}-\frac{1-\frac{1}{m_i}}{z-z_i}+\frac{c_i}{1-\frac{1}{m_i}}+\smallO(1) & z \rightarrow z_i, \\ -\frac{1}{z-z_j}\left(1+\left(\log \left|\frac{z-z_j}{J_1^{(j)}}\right|\right)^{-1}\right)+c_j+\smallO(1) & z \rightarrow z_j \\ -\frac{1}{z}\left(1+\left(\log \left|\frac{z}{J_{-1}^{(n)}}\right|\right)^{-1}\right)-\frac{c_n}{z^2}+\smallO\left(\frac{1}{|z|^2}\right), & z \rightarrow \infty\end{cases}
	\end{aligned}
	\label{derasymp}
	\end{equation}
	where $c_j$ are accessory parameters. Moreover, the $J_{1}^{(i)}$ and $J_{1}^{(j)}$  are the first coefficients in the following expansion of Hauptmodule $J$ in a neighborhood of each elliptic point with ramification index $m_i$ at $z=z_i$ and each puncture at  $z=z_j$, respectively, 
	\begin{equation*}
	\begin{aligned}
	J(z) =z_i + \sum_{k=1}^{\infty} J_k^{(i)} \Big(
	\frac{z-z_i}{z-\bar{z_i}}
	\Big),
	\end{aligned}
	\end{equation*}
	and
	\begin{equation}
	\begin{aligned}
	J(z) = z_j + \sum_{k=1}^{\infty} J_k^{(j)} \exp \Big(
	-\frac{2\pi \sqrt{-1} k}{|\delta_j| (z-z_j)}
	\Big).\label{Jparabolic1}
	\end{aligned}
	\end{equation}
	Additionally, in a neighborhood of the puncture at infinity, $z_n = \infty$, the $J_{-1}^{(n)}$ appears in the following expansion:
	\begin{equation}
	\begin{aligned}
	J(z) = \sum_{k=-1}^{\infty}
	J_k^{(n)} \exp\Big(
	\frac{2\pi\sqrt{-1}kz}{|\delta_n|}
	\Big).\label{Jparabolic2}
	\end{aligned}
	\end{equation}
	In equations \eqref{Jparabolic1} and \eqref{Jparabolic2}, the $\delta_{n_e+1}, \dots, \delta_{n} \in \mathbb{R}$ are called translation lengths of the associated parabolic generators.
	\section{Schottky uniformization of orbifold Riemann surfaces}\label{Schotkkyreview}
	In this appendix, we review some important facts about the Schottky uniformization of orbifold Riemann surfaces (see \cite{Taghavi2024classical} for more details). We'll begin firstly by recalling how a compact Riemann surface of genus $g\geq2$ is uniformized by a Schottky group and then extend it to the orbifold Riemann surface. Let's start with some well-known definitions.   
	
	A Schottky group $\Sigma$ with the limit set  $\Lambda$ is a discrete subgroup of M\"{o}bius group $\PSLC$ that acts properly on the region of discontinuity $\Omega = \hat{\cmpx}\backslash\Lambda$ of the Riemann sphere $\hat{\cmpx}$.\footnote{More precisely, it is strictly a loxodromic Kleinian group, which is also free and finitely generated \cite{Maskit_1967}.} Additionally, a Schottky group, $\Sigma$ of rank $g$, is referred to as marked when a relation-free set of generators $L_1, \dots, L_g  \in \PSLC$ is selected. A key result is that for every marked Schottky group $(\Sigma;L_1,\dots,L_g)$, there exists a fundamental domain 
	$\SchottkyFund$.\footnote{However, it is important to note that this fundamental domain $\SchottkyFund$ is not uniquely determined by the chosen marking of the Schottky group $\Sigma$.} This domain is a connected region in $\hat{\cmpx}$, bounded by $2g$ disjoint Jordan curves $\mathcal{C}_1,...\mathcal{C}_g,\mathcal{C}'_{1},...,\mathcal{C}'_g$, where 
	$\mathcal{C}'_i= -L_i(\mathcal{C}_i)$ for $i=1,...,g$. The orientations of $\mathcal{C}_i$
	and $\mathcal{C}'_i$ are opposite and they correspond to components of $\partial\SchottkyFund$. Accordingly, a compact Riemann surface can be constructed by $\Omega\slash\Sigma$. There is also a concept of equivalence between two marked Schottky groups: 
	$(\Sigma; L_1, \dots, L_g)$ is considered equivalent to 
	$(\tilde{\Sigma}; \tilde{L}_1, \dots, \tilde{L}_g)$ if there exists a Mobius transformation 
	$\varsigma \in \PSLC$ such that 
	$\tilde{L}_i = \varsigma  L_i \varsigma^{-1}$ for all $i=1,\dots,g$.  The collection of equivalence classes of marked Schottky groups of genus $g$ is known as the Schottky space of genus $g$ and is denoted by $\schottky_{g}$. The standard representation for each $L_i$ is given by the equation
	\begin{equation}
	\bigg(L_i (z) - a_i\bigg)(z-b_i) = \lambda_i\hspace{1mm}\bigg(L_i (z) - b_i\bigg)(z-a_i), \qquad z \in \hat{\cmpx},
	\end{equation}
	where $a_i$ and $b_i$ represent the attracting and repelling fixed points of the loxodromic element $L_i$ and $0 < |\lambda_i| < 1$ is the associated multiplier. Furthermore, using the notion of equivalence, one can set the attracting fixed points of generators $L_1$ and $L_2$ along with the repelling fixed-point of generator $L_1$,  to 0, 1, and $\infty$ respectively, thereby defining the normalized marked Schottky group.
	
	In the presence of cusps and conical singularities, 
	one can subtract from $\Omega$ the pre-images of cusps by the covering map $\Omega \to \hat{X}_O$ (where $\hat{X}_O$ is compactified underlying Riemann surface) to get another planar region $\Omega_0$. The region  $\Omega_0$ provides the uniformization for the underlying Riemann surface  $X_O$ and lets $\pi_{0}$ denote the corresponding covering map $\Omega_0 \to X_O \cong \Omega_0/\Sigma$. Next, we can lift the branch divisor $\brdiv$
	using the covering map $\pi_{0}$, resulting in a new branch divisor
	\begin{equation}
	\widetilde{\brdiv} := \sum_{z_i \in \pi_{0}^{-1}(\sing_{\curlywedge}(O))} \left(1-\frac{1}{\nu\left(\pi_{\Sigma}(z_i)\right)}\right) z_i,
	\end{equation}
	which resides in the planar region $\Omega_0$. Here, $\sing_{\curlywedge}(O)$ is the set of singular points of finite order for an orbifold $O$ and $\nu$ is branching function which assigns to each singular point its corresponding branching order, i.e. $\nu: \sing_{\curlywedge}(O) \to \orderrange :=\big(\mathbb{N}\backslash\{1\}\big) $.  Then, the pair $(\Omega_0,\widetilde{\brdiv})$ will define a planar Riemann orbisurface $\singrigon$ such that $\pi_{\Sigma}:\singrigon \to O \cong \singrigon/\Sigma$ serves as an orbifold covering map (see \cite{Wong-1971} for more details). 
	Now, we define $\singfund$ as the pair $(\SchottkyFund_0 , \widetilde{\brdiv}|_{_{\SchottkyFund}})$, where $\SchottkyFund_0$ is given by $\SchottkyFund \cap \Omega_0$, and $\widetilde{\brdiv}|_{_{\SchottkyFund}}$ represents the restriction of $\widetilde{\brdiv}$ to $\SchottkyFund$. Accordingly, the regularized singular fundamental domain of a Schottky group is defined by\footnote{Note that the restriction of $\pi_0$ to $\Omega^{\text{reg}} := \Omega_0 \backslash \operatorname{Supp}\widetilde{\brdiv}$ endows $X_O^{\text{reg}}$ with the Schottky global coordinate, such that the space of singular conformal metrics  $\mathscr{CM}(O)$ consisting of functions $\varphi$, satisfying the condition
		\begin{equation}\label{Schottkycovariant}
		\varphi \circ \gamma + \log|\gamma'|^2 = \varphi \quad \text{for all} \quad \gamma \in \Sigma,\nonumber
		\end{equation}
		and representing $\widetilde{\brdiv}$.}
	\begin{equation}
	\singfund_{\delta} = \singfund \big\backslash  \bigcup_{j=1}^{n}\left\{z \, \Big| \, |z-z_i|<\delta\right\}.
	\end{equation} 
	In this way, the generalized Schottky space $\schottky_{g,n}(\boldsymbol{m})$ of Riemann orbisurface $O$, with or without punctures, is viewed as a holomorphic fibration  $\jmath : \schottky_{g,n}(\boldsymbol{m}) \to \schottky_{g}$, where the fibers correspond to configuration spaces of $n$ labeled points with orders specified by the vector $\boldsymbol{m} = (m_1,\dots,m_n)$.\footnote{For details, see subsection 3.3 in \cite{Taghavi2024classical}.} Equivalently, with each $O \cong \singrigon/\Sigma$, we assign a point in the generalized Schottky space $\schottky_{g,n}(\boldsymbol{m})$.
	\section{Details of (co)Homology Double Complexes}\label{hocohomdet}
	In this appendix, we first review the association of double (co)homology complexes to a manifold $M$, which is quotient manifold $\tilde{\Omega}/\Gamma$, and $\Gamma$ acts properly discontinuously on it. Detailed explanations can be found in \cite{Aldrovandi_1997} and the references therein (also see \cite{Takhtajan_2003, park2017potentials, maclane2012homology}). Next, we adapt this framework to our specific scenario involving a 3-orbifold with lines of conical singularities and punctures where its conformal boundary is given by an orbifold Riemann surface.

	Utilizing the Uniformization theorem, a natural structure (function)  is induced on manifold $M$ from its covering space. Therefore, a good starting point for systematically defining that function (for example volume or Liouville action) is considering their planar covering spaces' homology and cohomology complexes. These complexes provide us with regions of integration and differential forms to define that function as a pairing between them.
	Specifically, the homology double complex $\compfontdii{K}$ is defined as the tensor product of the standard singular chain complex of $\tilde{\Omega}$ and the canonical bar-resolution complex for $\Gamma$, taken over the integral group ring $\mathbb{Z}\Gamma$.\footnote{The standard singular chain complex encodes topological information about space via its singular simplices, while the canonical bar-resolution complex offers algebraic data that reflects the structure of the group. Moreover,  $\mathbb{Z}\Gamma$ denotes the integral group ring associated with the group $\Gamma$, which consists of finite sums of the form $\sum_{\gamma\in\Gamma} n_{\gamma}\hspace{1mm}\gamma$, where each $n_{\gamma}$ is an integer. } 
	Meanwhile, the cohomology double complex  $\compfontuii{C}$ corresponds to the tensor product of bar-de Rham complex on
	$\tilde{\Omega}$ and group cohomology complex. Each of these concepts will be explained in more detail as follows.
	
	Let denote the standard singular chain complex of $\tilde{\Omega}$ and its differential, respectively, with $\compfontdi{S}\equiv \compfontdi{S}(\tilde{\Omega})$ and $\partial'$. The group $\Gamma$ acts on $\tilde{\Omega}$ by linear fractional transformations and induces a left action on $\compfontdi{S}$. Hence $\compfontdi{S}$ becomes a complex of  $\Gamma$-modules. Because the action of $\Gamma$  on $\tilde{\Omega}$ is proper,\footnote{An action of a group on a space is said to be proper if, for every compact subset of the space, only finitely many group elements map points within the subset to points outside it.} the complex $\compfontdi{S}$ can be viewed as a complex of left free $\mathbb{Z}\Gamma$-modules. Moreover, let us denote the canonical bar resolution complex for $\Gamma$ and its differential, respectively, by
	$\compfontdi{B}\equiv \compfontdi{B}(\mathbb{Z}\Gamma)$ and $\partial''$.  Each $\compfont{B}_{b}$ is a free left $\Gamma$-module on generators $[\gamma_1|...|\gamma_b]$, with the differential $\partial'' {:}\hspace{.5mm}\compfont{B}_b\rightarrow \compfont{B}_{b-1}$ given by
	\begin{equation}
	\begin{split}
	&\text{For}~ b=0:~~~\partial''[~] = 0,\\
	&\text{For}~ b=1:~~~\partial''[\gamma] = \gamma[~]-[~],\\
	&\text{For}~b>1:~~~\partial'' [\gamma_1|...|\gamma_b] = \gamma_{1} [\gamma_2|...|\gamma_b]+\sum_{j=1}^{b-1}(-1)^j [\gamma_1|...|\gamma_{j}\gamma_{j+1}|...|\gamma_b]+(-1)^b [\gamma_1|...|\gamma_{b-1}].\nonumber
	\end{split}
	\end{equation}
	If any of the group elements inside $[~]$ equals the unit element 1 in $\Gamma$ then $[\gamma_{1}|...|\gamma_{b}]$ is defined to be zero. Moreover, $\compfont{B}_{0}$ is a $\mathbb{Z}\Gamma$-module on one generator $[~]$, and can be identified with $\mathbb{Z}\Gamma$ under the isomorphism that sends $[~]$ to 1. Then, the double complex $\compfontdii{K}$ is defined as $\compfontdii{K}=\compfontdi{S}\otimes_{\mathbb{Z}\Gamma}\compfontdi{B}$. The associated total complex, denoted $\text{Tot}\compfont{K}$, is equipped with a total differential given by $\partial=\partial'+(-1)^a \partial''$ acting on $\compfont{K}_{a,b}$.
	It is important to note that since both $\compfontdi{S}$ and $\compfontdi{B}$ are complexes of left $\Gamma$-modules, we must provide each $\compfont{S}_a$ with a right $\Gamma$-module structure to properly define their tensor product over $\mathbb{Z}\Gamma$. This is achieved in the standard manner by setting the action as $c.\gamma= \gamma^{-1}(c)$ for any element $c$ in the module.
	In this way, $\compfontdi{S}$ is equivalent to $\compfontdi{S}\otimes_{\mathbb{Z}\Gamma}\compfont{B}_{0}$ under
	the correspondence $c\mapsto c\otimes[~]$.
	
	The statements above can be summarized as illustrated below.
	
	\[\begin{tikzcd}
	& \vdots & \vdots & \vdots \\
	0 & {\compfont{K}_{2,a}(\tilde{\Omega},\Gamma)} & {\compfont{K}_{1,a}(\tilde{\Omega},\Gamma)} & {\compfont{K}_{0,a}(\tilde{\Omega},\Gamma)} & 0 \\
	0 & {\compfont{K}_{2,a-1}(\tilde{\Omega},\Gamma)} & {\compfont{K}_{1,a-1}(\tilde{\Omega},\Gamma)} & {\compfont{K}_{0,a-1}(\tilde{\Omega},\Gamma)} & 0 \\
	& \vdots & \vdots & \vdots
	\arrow["{\partial''_{a+1}}", from=1-2, to=2-2]
	\arrow["{\partial''_{a+1}}", from=1-3, to=2-3]
	\arrow["{\partial''_{a+1}}", from=1-4, to=2-4]
	\arrow["i", hook, from=2-1, to=2-2]
	\arrow["{\partial'_2}", from=2-2, to=2-3]
	\arrow["{\partial''_{a}}", from=2-2, to=3-2]
	\arrow["{\partial'_1}", from=2-3, to=2-4]
	\arrow["{\partial''_{a}}", from=2-3, to=3-3]
	\arrow["{\partial'_0}", from=2-4, to=2-5]
	\arrow["{\partial''_{a}}", from=2-4, to=3-4]
	\arrow["i", hook, from=3-1, to=3-2]
	\arrow["{\partial'_2}", from=3-2, to=3-3]
	\arrow["{\partial''_{a-1}}", from=3-2, to=4-2]
	\arrow["{\partial'_1}", from=3-3, to=3-4]
	\arrow["{\partial''_{a-1}}", from=3-3, to=4-3]
	\arrow["{\partial'_0}", from=3-4, to=3-5]
	\arrow["{\partial''_{a-1}}", from=3-4, to=4-4]
	\end{tikzcd}\]
	
	\vspace{0.8em}
	In cohomology, the corresponding double complex is defined as follows. Let $\compfontui{A}\equiv\compfontui{A}_{\mathbb{C}}(\tilde{\Omega})$ represents the complexified de Rham complex on $\tilde{\Omega}$. Each $\compfont{A}^a$ is a left $\Gamma$-module with the pullback action of $\Gamma$, i.e., $\gamma.\hspace{.5mm}\tilde{\omega}=(\gamma^{-1})^{*}\tilde{\omega}$  for each $\tilde{\omega}\in \compfontui{A}$ and $\gamma\in \Gamma$. Define the double complex $\compfont{C}^{a,b}=\text{Hom}_{\mathbb{C}}(\compfont{B}_{b},\compfont{A}^a)$ with differential $\dd$, the usual de Rham differential, and $\delta=(\partial'')^{*}$, the group coboundary. For $\tilde{\omega}\in \compfont{C}^{a,b}$,
	\begin{equation}
	(\delta\tilde{\omega})_{\gamma_1,...,\gamma_{b+1}} = \gamma_1. \tilde{\omega}_{\gamma_1,...,\gamma_{b+1}}+\sum_{j=1}^{b}(-1)^{j} \tilde\omega_{\gamma_1,...,\gamma_{j}\gamma_{j+1},...,\gamma_{b+1}}+(-1)^{b+1}\tilde\omega_{\gamma_1,...,\gamma_b},\label{deltatildew}
	\end{equation}
	and the total differential on $\compfont{C}^{a,b}$ is denoted by $D=\dd+(-1)^{a}\delta$.
	
	These statements can also be summarized as illustrated below.
	
	\[\begin{tikzcd}
	& \vdots & \vdots & \vdots \\
	0 & {\compfont{C}^{2,a}(\tilde{\Omega},\Gamma)} & {\compfont{C}^{1,a}(\tilde{\Omega},\Gamma)} & {\compfont{C}^{0,a}(\tilde{\Omega},\Gamma)} & 0 \\
	0 & {\compfont{C}^{2,a-1}(\tilde{\Omega},\Gamma)} & {\compfont{C}^{1,a-1}(\tilde{\Omega},\Gamma)} & {\compfont{C}^{0,a-1}(\tilde{\Omega},\Gamma)} & 0 \\
	& \vdots & \vdots & \vdots
	\arrow["\delta", from=1-2, to=2-2]
	\arrow["\delta", from=1-3, to=2-3]
	\arrow["\delta", from=1-4, to=2-4]
	\arrow["i", from=2-1, to=2-2]
	\arrow["d", from=2-2, to=2-3]
	\arrow["\delta", from=2-2, to=3-2]
	\arrow["d", from=2-3, to=2-4]
	\arrow["\delta", from=2-3, to=3-3]
	\arrow["d", from=2-4, to=2-5]
	\arrow["\delta", from=2-4, to=3-4]
	\arrow["i", hook, from=3-1, to=3-2]
	\arrow["d", from=3-2, to=3-3]
	\arrow["\delta", from=3-2, to=4-2]
	\arrow["d", from=3-3, to=3-4]
	\arrow["\delta", from=3-3, to=4-3]
	\arrow["d", from=3-4, to=3-5]
	\arrow["\delta", from=3-4, to=4-4]
	\end{tikzcd}\]

	A natural pairing exists between $\compfont{C}^{a,b}$ and $\compfont{K}_{a,b}$ which assigns to the pair $(\tilde{\omega},c\hspace{.1mm}\otimes[\gamma_1|...|\gamma_{b}])$ the evolution of the $a$-form $\tilde{\omega}_{\gamma_1,...,\gamma_b}$ over the $a$-cycle $c$,
	\begin{equation}
	\stks{\tilde{\omega},c\otimes[\gamma_1|...|\gamma_b]} = \int_{c} \tilde{\omega}_{\gamma_1,...,\gamma_b}.\label{pairing}
	\end{equation}
	An important point to note is that for the singular homology of $\tilde{\Omega}$, the group homology of $\Gamma$, and the homology of the complex $(\text{Tot}~\compfont{K})_\bullet$, there are isomorphisms:
	\begin{equation}
	H_{\bullet}(M,\mathbb{Z}) \cong H_{\bullet}(\Gamma,\mathbb{Z})\cong H_{\bullet}((\text{Tot}~\compfont{K})_\bullet).\label{rep1}
	\end{equation}
	Similarly, for the de Rham cohomology of $\tilde{\Omega}$, the group cohomology of $\Gamma$, and the cohomology of the complex $(\text{Tot}~\compfont{C})^\bullet$, there are also isomorphisms:
	\begin{equation}
	H^{\bullet}(M,\mathbb{C}) \cong H^{\bullet}(\Gamma,\mathbb{C}) \cong H^{\bullet}((\text{Tot}~\compfont{C})^\bullet).\label{rep2}
	\end{equation}
	These isomorphisms imply that the pairing \eqref{pairing} is a non-degenerate pairing between the corresponding cohomology and homology groups $H^{\bullet}((\text{Tot}\compfont{C})^\bullet)$ and $H_{\bullet}((\text{Tot}\compfont{K})_\bullet)$, due to de‌ Rham theorem. Specifically, if $\tilde\omega$ is a cocycle in $(\text{Tot}\compfont{C})^{\bullet}$ and $c$ is a cycle in $(\text{Tot}\compfont{K})_\bullet$, the pairing $\stks{\tilde\omega,c}$ depends only on the cohomology class $[\tilde\omega]$ and the homology class $[c]$, not on their individual representatives. Accordingly, building the volume or Liouville action from the pairing \eqref{pairing} implies that these functions do not depend on the choice of a fundamental domain.
	
	The Stokes theorem, combined with the isomorphisms in equations \eqref{rep1} and \eqref{rep2}, implies that
	\begin{equation}
	\stks{\delta\tilde\omega,c}=\stks{\tilde\omega,\partial'' c},
	\end{equation}
	and, accordingly,
	\begin{equation}
	\stks{D\tilde\omega,c}=\stks{\tilde\omega,\partial c}.\label{Dpartial}
	\end{equation}
	\subsection{Homology in bulk: homology group and group homology}\label{homology}
	In this subsection, we delve into the specifics of the double homology complex of bulk extension of $\Omega$ in the presence of singularities (both conical and puncture points) within the context of Schottky uniformization. To aid in understanding, in the following, it's useful to visualize the case for genus  $g=2$ --- refer to Figures \ref{fig:1} and \ref{fig:2} for clarity.
	\begin{figure}[h]
		\centering
		\includegraphics[width=16em]{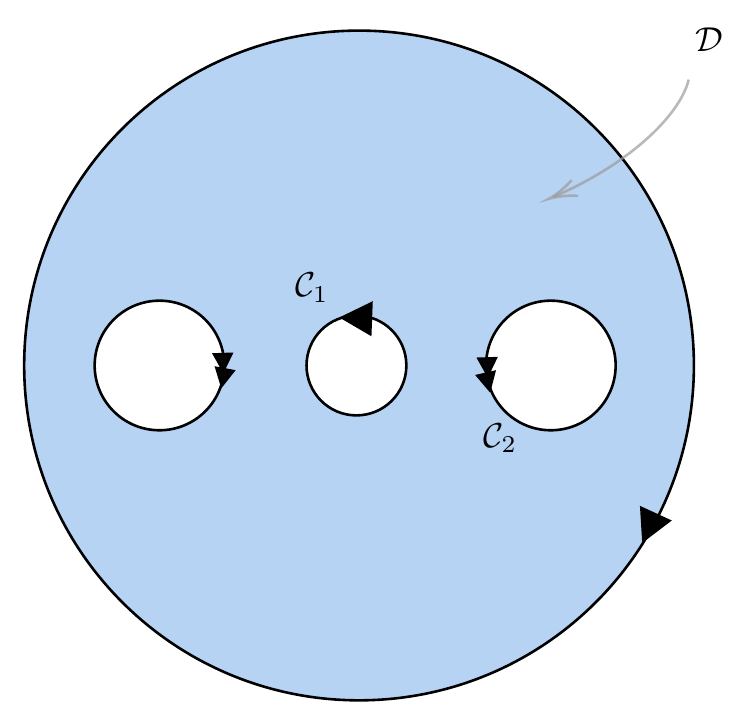}
		\caption{Fundametnal domain $\SchottkyFund$ of Schottky uniformization for $g=2$.}
		\label{fig:1}
	\end{figure}
	\begin{figure}[h]
		\centering
		\includegraphics[width=22em]{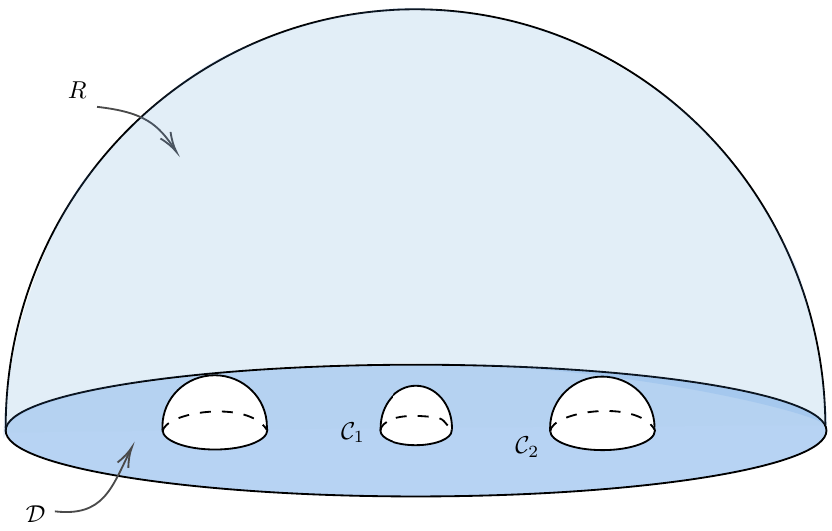}
		\caption{The three dimensional extension of fundamental domain $\SchottkyFund$ in Figure \ref{fig:1}.}
		\label{fig:2}
	\end{figure}
	
	First, let's review the homological construction in the absence of singularities. Consider a finitely generated purely loxodromic subgroup $\Sigma$ of $\PSLC$, namely a marked  Schottky group with a region of discontinuity $\Omega \subset\hat{\mathbb{C}}$ (see \cite{zograf1988uniformization,Taghavi2024classical} and Appendix \ref{Schotkkyreview} for more details). The compact Riemann surface $X =\Omega\slash \Sigma$ is the conformal boundary of the corresponding hyperbolic three-manifold  $M = (\mathbb{U}^3\cup \Omega)\slash\Sigma$, where $\mathbb{U}^3 = \{(z,r) \; | \; z\in \mathbb{C}, r>0
	\}$. Accordingly, the fundamental domain $\SchottkyFund$ can be extended to the upper-half space with a geodesic representative, as is shown in Figure \ref{fig:2}.
	
	Define $\tilde\Omega \equiv \mathbb{U}^3\cup \Omega $, and let $\compfontdi{S}=\compfontdi{S}({\tilde\Omega})$ and $\compfontdi{B}=\compfontdi{B}(\mathbb{Z}\Sigma)$ denote the singular chain complex of  $\tilde\Omega$ and standard bar resolution complex for $\Sigma$, respectively. The key idea is to construct a 3-chain that represents the fundamental region (the blue region in Figure \ref{fig:2}) in total complex $\text{Tot} \compfont{K}$ of the double homology complex $\compfontdii{K}= \compfontdi{S} \otimes_{\mathbb{Z}\Sigma} \compfontdi{B}$, as follows. Identify $R\subset \mathbb{U}^3$ with $R\otimes[\;] \in \compfont{K}_{3,0}$ as the fundamental region of $\Sigma$ in $\tilde\Omega$ such that $R\cap\Omega$ corresponds to the fundamental domain $\SchottkyFund$.
	
	We proceed to construct a staircase of homology elements and then contract them with suitable cohomology elements (made in the next subsection) to define the desired action. This process starts with noting $\partial'' R =0$, and
	
	\begin{equation}
	\partial' R = -\SchottkyFund- \sum_{i=1}^{g} \big(H_i -L_i(H_i)\big) = -\SchottkyFund + \partial'' S,\label{pprimeR}
	\end{equation}
	where $H_i$ is a topological hemisphere and $S \in \compfont{K}_{2,1}$ is defined by
	\begin{equation}
	S = \sum_{i=1}^{g} H_i \otimes [L_i^{-1}].\label{}
	\end{equation}
	We can proceed similarly and obtain
	\begin{equation}
	\begin{aligned}
	\partial'' S = -
	\sum_{i=1}^{g} \big(H_i\otimes [\;] - L_i(H_i) \otimes [\;]\big),
	\end{aligned}
	\end{equation}
	and
	\begin{equation}
	\begin{aligned}
	\partial'S = L,\label{pprimeS}
	\end{aligned}
	\end{equation}
	where $L= \sum_{i=1}^{g} \mathcal{C}_i \otimes [L_i^{-1}] \in \compfont{K}_{1,1}$ and $\mathcal{C}_i$ indicates Schottky circles shown in Figure \ref{fig:2}. Since the total complex $\text{Tot} \compfont{K}$ is equipped with the total differential $\partial= \partial'+(-1)^a \partial''$ on $\compfont{K}_{a,b}$, then the 3-chain $R-S \in (\text{Tot}\compfont{K})_3$ satisfies\footnote{It is straightforward to see that $\partial\hspace{.4mm}\Xi=0$ and if $\SchottkyFund$ and $\tilde{\SchottkyFund}$ are two choices of the fundamental domain for $\Sigma$ in $\Omega$, then $[\Xi]=[\tilde{\Xi}]$ for the corresponding classes in $H_2(\text{Tot}~\compfont{K})_2$.} (see equations \eqref{pprimeR} and \eqref{pprimeS})
	\begin{equation}
	\partial (R-S) = -\SchottkyFund -L\overset{\text{def}}{=}-\Xi.\label{RmS}
	\end{equation} 
	As discussed in Section~\ref{Renvolume}, it is important to note that defining the renormalized volume requires introducing a regularizing surface $f=\varepsilon$ near the conformal boundary. Consequently, the fundamental region $R$, should be modified according to $ R_\varepsilon = R\cap \{f\geq \varepsilon\}$. Additionally, as $r$ or $ \varepsilon$ approaches $0$, the cycle 
	$\Xi$ on $f=\varepsilon$
	serves as an extension of the fundamental domain $\SchottkyFund$ (see \cite{Aldrovandi_1997}), therefore $R-S$ remains a consistent extension of $R$.
	
	\begin{figure}[h]
		\centering
		\includegraphics[width=\textwidth]{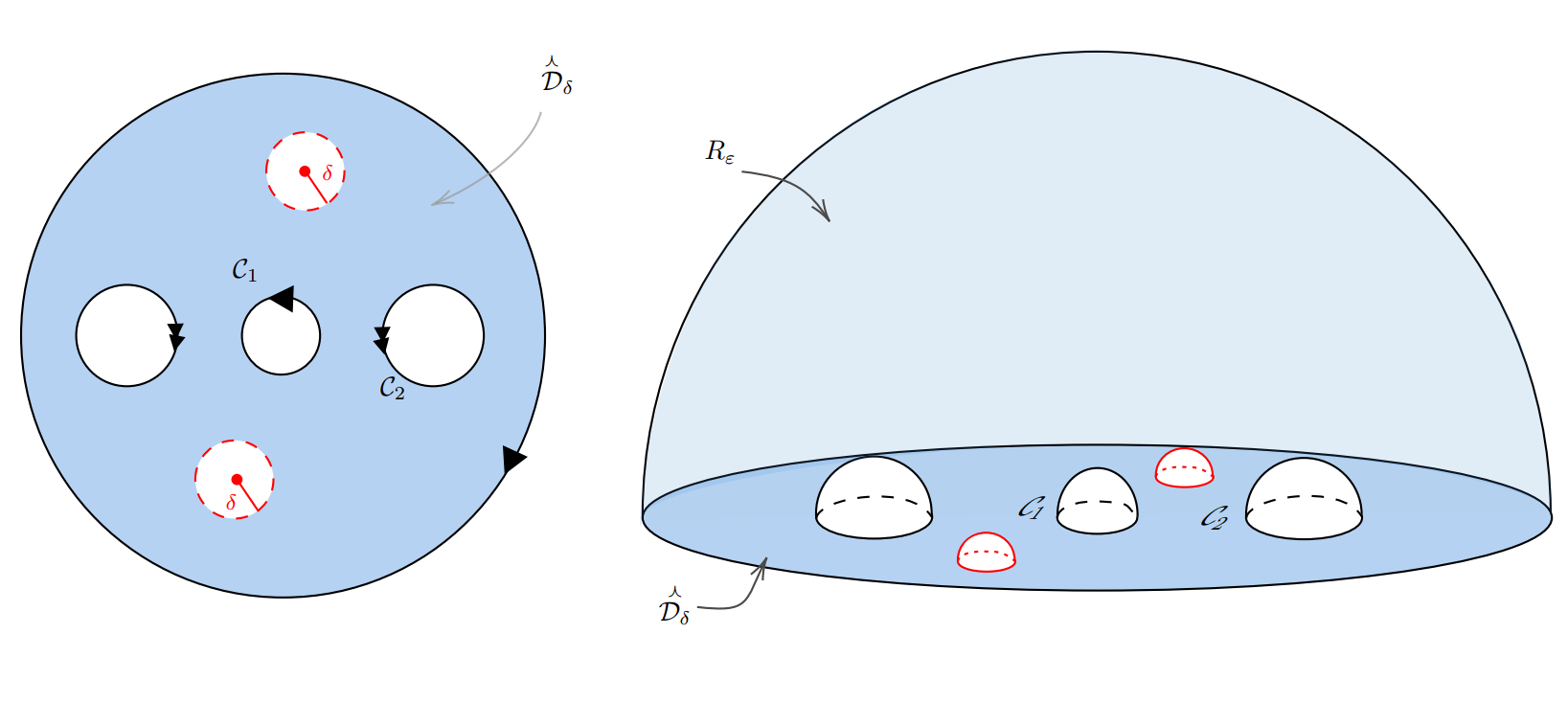}
		\caption{The three-dimensional extension of fundamental domain with defect points. The left figure is its fundamental domain $\singfund_{\delta}$, and the right figure shows the geodesic extension to $\mathbb{U}^3$.}
		\label{Puncturedfig}
	\end{figure}
	
	By adding punctures and conical points to the fundamental domain $\SchottkyFund$, we cut spheres around them, as shown in Figure \ref{Puncturedfig}. As demonstrated in Section~\ref{Renvolume} (see also \cite{park2017potentials}), the regularizing surface $f=\varepsilon$   intersects the singularities, therefore to have a proper level-defining function, we also need to remove the neighborhoods of each singularity at $(z_i,0)$ in $\mathbb{U}^3$.
	\begin{figure}[h]
		\centering
		\includegraphics[width=27em]{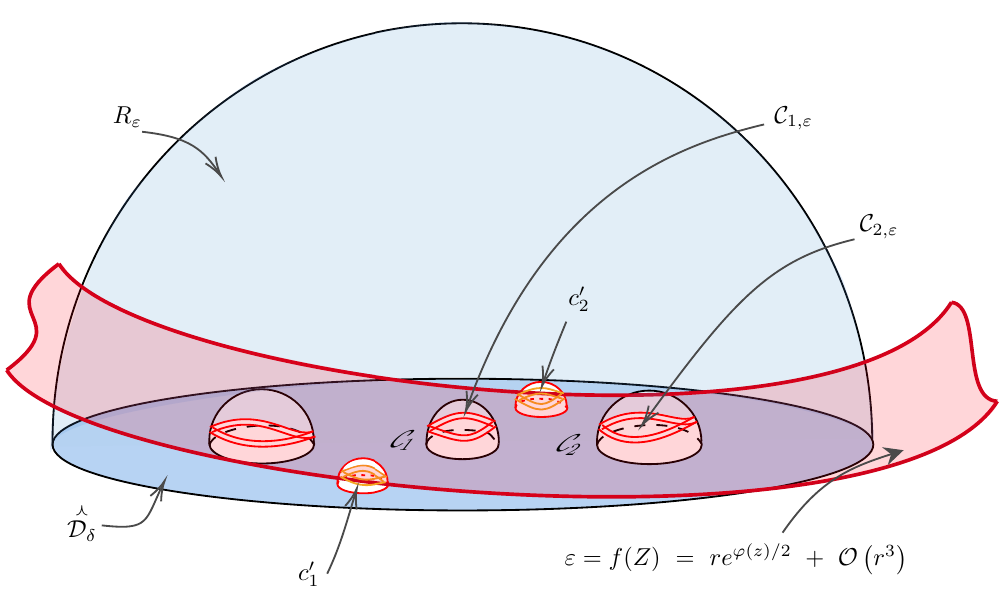}
		\caption{The regularizing surface $f(Z)$ which cuts through the fundamental region. }
		\label{fig:3}
	\end{figure}
	For $\varepsilon>0$, let
	\begin{equation}
	R_{\varepsilon} = R\cap\{f\geq\varepsilon\}\backslash \bigcup_{j=1}^{n}\big{\{}(z,r)\in\mathbb{U}^3\big| ~||(z,r)-(z_i,0)|| \leq \tilde{\varepsilon}\big{\}},
	\end{equation}
	be the regularized truncated fundamental region. For $H_{i,\varepsilon}= H_i\cap R_{\varepsilon}$, the equation \eqref{pprimeR} changes to
	\begin{equation}
	\partial' R_\varepsilon = -\singfund_{\varepsilon}- \sum_{i=1}^{g} \big(H_{i,\varepsilon} -L_i(H_{i,\varepsilon})\big) = -\singfund_{\varepsilon} + \partial'' S_{\varepsilon},\label{pprimeReps}
	\end{equation}
	with 
	\begin{equation}
	S_\varepsilon = \sum_{i=1}^{g} H_{i,\varepsilon} \otimes [L_i^{-1}].
	\end{equation}
	Furthermore,\footnote{Note that the $\mathcal{C}_{i,\varepsilon}$ represents the intersection of $H_{i,\varepsilon}$ with the regularizing surface $f = \varepsilon$.}  
	\begin{equation}
	\partial'S_\varepsilon  = \sum_{i=1}^{g} \mathcal{C}_{i,\varepsilon} \otimes [L_i^{-1}] = L,\label{pprimeSeps}
	\end{equation}
	and 
	\begin{equation}
	\partial' \singfund_{\varepsilon} 
	=\partial'' L  + \sum_{j=1}^{n} c'_j\otimes [\;],\label{parD}
	\end{equation}
	where $c_j'$s are the truncated circles around puncture and conical points shown in Figure \ref{fig:3}. Moreover, instead of \eqref{RmS}, by using equations \eqref{pprimeReps} and \eqref{pprimeSeps} one can see
	\begin{equation}
	\partial (R_{\varepsilon}-S_{\varepsilon}) = -\singfund_{\varepsilon}  -L\overset{\text{def}}{=}-\Xi_{\varepsilon}.\label{RmSeps}
	\end{equation}

	\subsection{Cohomology in bulk: cohomology group and group cohomology}\label{CohomologyU3}
	Let us now turn to the double cohomology complex of bulk extension of $\Omega$ in the presence of singularities ( punctures and conical points) within the context of Schottky uniformization. The cohomology construction starts with the volume form on this space, $\omega_3$,
	\begin{equation}
	\omega_3 = \frac{1}{r^3} ~\dd x\wedge \dd y\wedge \dd r= \frac{i}{2r^3}~\dd z\wedge \dd\bar{z}\wedge\dd r \in \compfont{C}^{3,0},
	\end{equation}
	It is an exact form since $\omega_3= \dd \omega_2$ with
	\begin{equation}
	\omega_2 = -\frac{i}{4r^2}~\dd z\wedge \dd\bar{z}\in \compfont{C}^{2,0}.\label{omega2}
	\end{equation}
	For the group element $\gamma = \begin{pmatrix}
	a & b \\
	c & d 
	\end{pmatrix}\in \Sigma\subset \PSLC$, a straightforward calculations using the equation \eqref{deltatildew} yields
	\begin{equation}
	\begin{split}
	\left(\delta\omega_2\right)_{\gamma^{-1}}&= \omega_2.\gamma-\omega_2 \\&
	=\frac{i}{2} J_{\gamma}(Z)\Bigg(\left|c\right|^2 \dd z\wedge\dd\bar{z}-\frac{c\overline{(c z+d)}}{r}\dd z\wedge\dd r+\frac{\bar{c}(cz+d)}{r}\dd\bar{z}\wedge\dd r\Bigg),
	\end{split}\label{deltaomega2}
	\end{equation}
	where $J_{\gamma}(Z)$ is defined in \eqref{Jgammaz}. It is important to note that in simplifying equation \eqref{deltaomega2}, we solely relied on the fact that $\text{det}(\gamma)=1$. Since 
	\begin{equation}
	\dd \delta\omega_2 =\delta \dd \omega_2 =\delta\omega_3,
	\end{equation}
	and $\omega_3$ is invariant under the  transformations by the group $\Sigma$, $\delta\omega_3=0$, this implies that there exists $\omega_1\in \compfont{C}^{1,1}$ such that $\delta\omega_2=\dd\omega_1$. Specifically,
	\begin{equation}
	\begin{aligned}
	\big(\omega_1\big)_{\gamma^{-1}} = -\frac{i}{8} \log \Big(
	|r~c(\gamma)|^2 J_\gamma(Z)
	\Big) \Big(
	\frac{\gamma''}{\gamma'}dz -  \frac{\overline{\gamma''}}{\overline{\gamma'}}d\overline{z}
	\Big),
	\end{aligned}
	\end{equation}
	where $c(\gamma)$ is the left-hand lower element in the matrix representation of the generator $\gamma$. Proceeding further to compute  $\delta\omega_1\in \compfont{C}^{1,2}$, we find that (see \cite{Takhtajan_2003})
	\begin{equation}
	\begin{split}
	\left(\delta\omega_1\right)_{\gamma_1^{-1},\gamma_{2}^{-1}} &= -\frac{i}{8}\log\left(J_{\gamma}(Z)~\frac{\left|c(\gamma_2)\right|^2}{\left|c(\gamma_2\gamma_1)\right|^2}\right)\left(\frac{\gamma_2''}{\gamma_2'}\circ\gamma_1\gamma_1' \dd z -\frac{\overline{\gamma_2''}}{\overline{\gamma_2'}} \circ\gamma_1\overline{\gamma_1'} \dd \bar{z}\right)\\&-\frac{i}{8}\log\left(J_{\gamma_2}(\gamma_1 Z)~\frac{\left|c(\gamma_2\gamma_1)\right|^2}{\left|c(\gamma_1)\right|^2}\right)\left(\frac{\overline{\gamma_1''}}{\overline{\gamma_1'}}~\dd \bar{z}-\frac{\gamma_1''}{\gamma_1'}~\dd z\right).\label{deltaomega1}
	\end{split}
	\end{equation}
	From the explicit calculations in equation \eqref{deltaomega1}, or simply by observing that
	\begin{equation}
	\dd \delta\omega_1 = \delta\dd\omega_1=\delta(\delta\omega_2)=0,
	\end{equation}
	we conclude that $\delta\omega_1$ is a closed. Therefore, there exists $\omega_0\in \compfont{C}^{0,2}$ such that $\delta\omega_1=\dd \omega_0$. Additionally, using $H^{3}({\tilde\Omega}, \Sigma)=0$, the antiderivative can be chosen such that $\delta\omega_0=0$. Finally, since the total complex $\text{Tot} \compfont{C}$ is equipped with the total differential $D= \dd+(-1)^a \delta$ on $\compfont{C}^{a,b}$, then for \begin{equation}
	\varpi= \omega_2-\omega_1-\omega_0\in (\text{Tot}\compfont{C})^2,\label{varpi}
	\end{equation}
	one gets
	\begin{equation}
	D\varpi = \omega_3.\label{Dvarpi}
	\end{equation}
	It is worth recalling that $\varpi$ represents an element of $H^{2}(X,\textcolor{black}{\singrigon})$, where $X$ is the regularizing surface $f=\varepsilon$, which, as $r\rightarrow 0$ , becomes the conformal boundary of the  orbifold $M$. Accordingly, in this process, we constructed the desired cohomology element, that its pairing with $\Xi_\varepsilon$, will give regularized Liouville action.
	\section{On existence and construction of regularizing surface $f(Z)=\varepsilon$}\label{regsurf}
	This appendix provides a detailed explanation of the construction of a well-defined regularizing surface, $f=\varepsilon$, within the region of discontinuity of the Schottky 3-orbifold $M$, accommodating cases both with and without conical singularities and punctures. The construction relies on a partition of unity for the group $\Gamma$ on  $\mathbb{U}^3\cup\Omega$, the region of discontinuity. A key step in this process is demonstrating the existence of a partition of unity, specifically in the region of discontinuity. This is achieved by first constructing the partition of unity on the orbifold itself. Then, if $\pi: \mathbb{U}^3\cup\Omega\rightarrow (\mathbb{U}^3\cup\Omega)\slash\Gamma$ denotes the natural projection map, by pulling back the partition of unity via $\pi^{-1}$, its existence on the region of discontinuity is established, as formalized in Lemma \ref{lemma} and demonstrated in figure \ref{fig:4}. Once the partition of unity on $\mathbb{U}^3\cup\Omega$ is established, it can be used to construct the $\Gamma$-automotphic regularizing surface  $f=\varepsilon$, which plays a pivotal role in the holographic description of (generalized) Liouville action in Section \ref{Renvolume}, based on the method outlined in \cite{Takhtajan_2003}.
	\begin{figure}[h]
		\centering
		\includegraphics[width=14em]{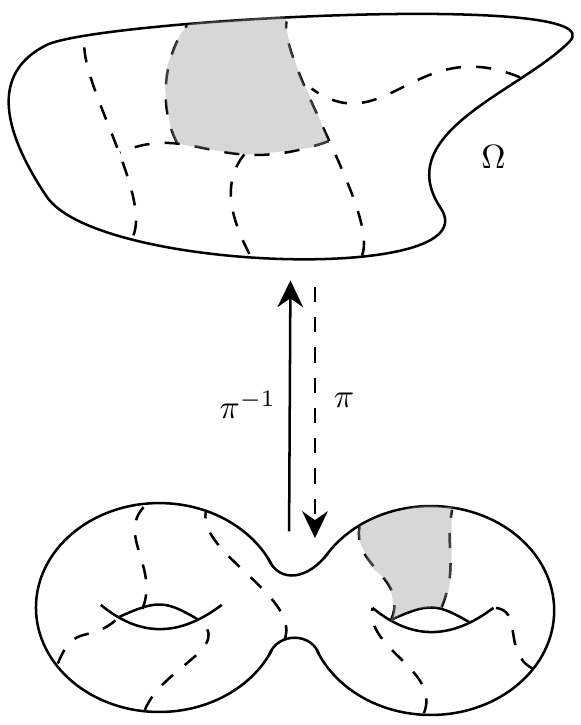}
		\caption{A sketch of pulling back partition of unity from 2-orbifold to $\Omega$ by inverse projection map $\pi^{-1}$. The same idea holds for 3D bulk.}
		\label{fig:4}
	\end{figure}
	
	Before presenting Lemma \ref{lemma}, we review essential definitions and facts about partitions of unity. The first step in constructing a partition of unity is associating a finite-indexed open cover to the orbifold. For a compact manifold the finite-indexed open cover is provided by definition. A topological space 
	$M$ is compact if every open cover of $M$ contains a finite subcover. That is, if  $M$ is covered by a collection of open sets $\{u_{\alpha}\}_{\alpha\in A}$
	, there exists a finite subset $\{u_{\alpha_1},u_{\alpha_2},...,u_{\alpha_n}\}$ that also covers $M$. In the presence of punctures, however, the orbifold is not compact. Instead, it is paracompact, 
	which allows for the construction of a locally finite open cover. 
	
	\textbf{Definition:}
	A topological space 
	$M$ is paracompact if every open cover has a locally finite refinement. 
	
	\textbf{Definition:} A collection of open sets  $\tilde{\mathcal{C}} = \{\tilde{u}_{\beta}\}_{\beta\in B}$  refines another collection $\mathcal{C} =  \{u_{\alpha}\}_{\alpha\in A}$ if each $\tilde{u}_{\beta}$ is contained within some $u_{\alpha}$. Alternatively, we say $\tilde{\mathcal{C}}$ is a refinement of $\mathcal{C}$. Also, the refinement is called locally finite if every point in $M$ has a neighborhood intersecting only finitely many sets $\tilde{u}_{\beta}$ in $ \tilde{\mathcal{C}}$.

	With the open covering specified, the existence of partition of unity is assured for compact manifold with respect to ${\mathcal{C}}$, and for paracompact orbifold with respect to $\tilde{\mathcal{C}}$.
	
	\noindent
	\textbf{Definition:}
	A partition of unity on a manifold $M$, subordinate to an open cover $\{u_\alpha\}_{\alpha \in A}$ consists of a family of smooth functions $\{\eta_{\alpha}: M\rightarrow [0,1]\}_{\alpha\in A}$ satisfying:
	\begin{enumerate}
		\item \textbf{Support on the cover}: Each $\eta_{\alpha}$  is supported within an open set $u_{\alpha}$. Specifically:
		\begin{equation}\nonumber
		\text{Supp}(\eta_{\alpha}) \subseteq u_{\alpha},
		\end{equation}
		where $0\leq\eta_{\alpha}\leq 1$.
		\item \textbf{Locally finite}: The family of supports is locally finite, which is to say, at any point $Z \in M$, only finitely many $\eta_{\alpha}(Z)$ are nonzero. 
		\item \textbf{Sum equals one}: The functions  $\eta_{\alpha}$ form a "partition" of unity, meaning:
		\begin{equation}\nonumber
		\sum_{\alpha\in A}\eta_{\alpha}(Z)=1,~~~\text{for all}~ Z\in M.
		\end{equation}
	\end{enumerate}
	The same notions and definitions apply to Schottky 3-manifolds with lines of conical singularities and punctures. So, by pulling back, its covering space, $\mathbb{U}^3 \cup\Omega$, is provided with a partition of unity (See Lemma 3.1 of \cite{kra1972}), which paves the way for defining a $\Gamma$-automorphic function on $\mathbb{U}^3 \cup\Omega$. 
	
	\begin{lemma}[I.kra \cite{kra1972}]\label{lemma}
		There exists a function $\eta\in C^{\infty}(\mathbb{U}^3\cup\Omega)$ satisfying:
		\begin{enumerate}
			\item $0\leq\eta\leq 1$,
			\item For each regular point $Z\in \mathbb{U}^3\cup\Omega$, there is a neighborhood $u$ of $Z$ and a finite subset $\tilde{\Gamma}\subset\Gamma$ such that $\eta|_{\gamma(u)}=0$ for each $\gamma\in \Gamma\backslash\tilde{\Gamma}$,
			\item $\sum_{\gamma\in\Gamma}\eta(\gamma Z)=1$,\hspace{.2cm}$Z\in\mathbb{U}^3\cup\Omega$.
		\end{enumerate}
		Additionally, the second property can be extended for puncture and conical points. If $R$ denotes a fundamental domain for  $\Gamma$ in $\mathbb{U}^{3}\cup\Omega$, then for each puncture and conical singularity on $(\mathbb{U}^{3}\cup\Omega)\slash\Gamma$, we can select regions, $\tilde{\mathcal{S}}_{\text{cusp}},\tilde{\mathcal{S}}_{\text{con}}$, within $R$ associated with each singularity.  These regions satisfy $\eta|_{\gamma (\tilde{\mathcal{S}}_{\text{cusp}})}=0$ for all $\gamma \in \Gamma \backslash\{1,\kappa\}$ and $\eta|_{\gamma (\tilde{\mathcal{S}}_{\text{con}})}=0$ for all $\gamma \in \Gamma \backslash\{1,\tau,\tau^2,...,\tau^{m-1}\}$, where $\kappa$ and $\tau$,  represents the parabolic and elliptic generators corresponding to the respective puncture and conical singularity. 
	\end{lemma}
	\begin{proof}
		The main idea is to pull back the partition of unity defined on the orbifold to the region of discontinuity in the covering space using $\pi^{-1}$ (as shown in figure \ref{fig:4}) and then use it to define $\eta$. Introducing an open cover begins with collecting a set of neighborhoods of each point. For each $Z\in (\mathbb{U}^{3}\cup\Omega)\slash\Gamma$, select a coordinate neighborhood $u(Z)$ of $Z$, chosen arbitrarily under the condition that the restriction of $\pi$ to each component of $\pi^{-1}(u(Z))$ is an $N(Z)$-to-one covering, where $N(Z)$ is the order of stability subgroup of $\pi^{-1}(Z)$.\footnote{The stability subgroup (also called the isotropy subgroup) of 
			$\pi^{-1}(Z)$ refers to the subgroup of $\Gamma$ that leaves the entire preimage 
			$\pi^{-1}(Z)$ invariant under the action of $\Gamma$.}
		
		In the presence of parabolic generators, additional considerations are necessary since the quotient space is non-compact. However, as noted earlier, this space (which we refer to as orbifold Riemann surface, or orbifold for short) is paracompact; that is, any open cover has a locally finite refinement by definition. Define puncture's neighborhoods as 
		$\mathcal{S}_{\text{cusp},i}$ (where $p_i\in \mathcal{S}_{\text{cusp},i}$) such that $\text{CL}~\mathcal{S}_{ \text{cusp},i}\subset\tilde{\mathcal{S}}_{\text{cusp},i}$, where CL denotes closure. If $\{p_1,...,p_n\}$ represents all punctures, define $\mathcal{S}_{ \text{cusp}}=\cup_{i}\mathcal{S}_{\text{cusp},i}$. For $Z \notin \text{CL}~\mathcal{S}_{\text{cusp}}$,\footnote{Note that $\text{CL}~\mathcal{S}=\cup_{i}\text{CL}~\mathcal{S}_i$.} require that $u(Z)\cap \text{CL}~\mathcal{S}_{\text{cusp}}=0$. If $Z\in \tilde{\mathcal{S}}_{\text{cusp},i}$ for some $i$ and $\pi^{-1}(Z)\in \text{CL}~R$, require that one component of $\pi^{-1}(u(Z))$ lies within $R\cup \kappa_i (R)$, where $\kappa_i$ is the parabolic generator corresponding to $p_i$. It is evident that such neighborhoods can be chosen since orbifold is a Hausdorff space. All these construction runs parallel for conical points; just replace
		$\mathcal{S}_{\text{cusp},i}$ with 
		$\mathcal{S}_{\text{con},i}$ and $R \cup \kappa_i(R)$ with
		$R \cup \tau_i(R) \cup \tau_i^2(R) \cup \dots \cup \tau_i^{m_i-1}(R)$. The procedure, actually, amounts to operating regular neighborhoods from singularities. With this setup, $\mathcal{U}=\{u(Z)| Z\in (\mathbb{U}^3\cup\Omega)\slash\Gamma\}$ forms an open cover and $(\mathbb{U}^3\cup\Omega)\slash\Gamma$ is  paracompact orbifold.
		
		Now, a locally finite subcover $\mathcal{U}_0= \{u(x_j), j=1,2,...\}$, can be chosen from the open covering $\mathcal{U}$. Let $\{\tilde{\eta}_j\}$ represent a smooth partition of unity subordinate to this subcover $\mathcal{U}_0$. Specifically, each $\tilde{\eta}_j$ is a smooth function on $(\mathbb{U}^3\cup\Omega)\slash\Gamma$ with support contained in $u(x_j)$ such that $\sum_{j}\tilde{\eta}_j(Z)=1$ for all $Z\in (\mathbb{U}^3\cup\Omega)\slash\Gamma$. For each $j$, select a single component $U(x_j)$ of $\pi^{-1}(u(x_j))$. If $x_j\in \tilde{\mathcal{S}}_{\text{cusp},i}$ for some $i$,
		it is required that $U(x_j) \subset R\cup\kappa_{i}(R)$. If $x_j\in \tilde{\mathcal{S}}_{\text{con},i}$ for some $i$,
		it is required that $U(x_j) \subset (R)\cup \tau_i(R) \cup \dots \cup \tau_i^{m_i-1}(R)$.

		The function $\tilde{\eta}_j$ is then mapped to the region of discontinuity as follows:
		\begin{equation}
		\eta_{j}(Z) = 	 \left\{
		\begin{split}
		&\tilde{\eta}_{j}(\pi(Z)), &\hspace{.5cm} Z\in U(x_j) ,\\
		&0  ~~~~~~, & Z\in (\mathbb{U}^3\cup\Omega)\backslash U(x_j).
		\end{split}
		\right.
		\end{equation} 
		It is evident that $\eta_{j}(Z)\in \mathcal{C}^{\infty}(\mathbb{U}^{3}\cup\Omega)$. Finally, a function satisfying all three conditions required for a partition of unity on the region of discontinuity is constructed as:
		\begin{equation}
		\eta(Z) = \sum_{j=1}^{\infty} \frac{\eta_{j}(Z)}{N(x_j)},~~~~Z\in \mathbb{U}^3\cup \Omega.
		\end{equation}
	\end{proof}
	Once the partition of unity, $\eta(Z)$, is established, the $\Gamma$-automorphic regularizing surface is defined as
	\begin{equation}
	f(Z) = \sum_{\gamma\in\Gamma}\eta(\gamma Z)\hat{f}(\gamma Z).\label{f}
	\end{equation}
	where
	\begin{equation}
	\hat{f} (Z) = 	 \left\{
	\begin{split}
	&re^{\varphi(z)/2}, &\hspace{.5cm} r\leq \varepsilon/2 ,\\
	&1  \qquad, & r\geq \varepsilon.
	\end{split}
	\right.\label{fhat}
	\end{equation}
	
	By property (2) of Lemma \ref{lemma}, for every $Z\in \mathbb{U}^3\cup\Omega$, the sum involves only finitely many nonzero terms, ensuring that the function $f$ is well-defined. Furthermore, properties (1) and (3) of Lemma \ref{lemma} guarantee that $f$ is positive on $\mathbb{U}^3\cup\Omega$. Note that the regularized action constructed with this $\Gamma$-automorphic regularizing surface is also invariant under $\Gamma$ action, which is why we employed it.
	
	We now determine the asymptotic behavior of  $f(Z)$ in \eqref{f}, as it is essential for computing the regularized bulk volume in Section \ref{Renvolume}. According to equation \eqref{adstr}, we find the following expressions for transformed coordinates under the action of $\gamma$,
	\begin{equation}
	\begin{aligned}
	& z(\gamma Z)=  \gamma(z) -\frac{1}{2} \gamma'(z)\hspace{1mm}\frac{\overline{\gamma''(z)}}{\overline{\gamma'(z)}}\hspace{1mm}r^2+ \mathcal{O}(r^4),\\
	&t(\gamma Z)=|\gamma'(z)|\hspace{1mm}r -\frac{1}{4}\frac{|\gamma''(z)|^2}{|\gamma'(z)|}\hspace{1mm}r^3 + \mathcal{O}(r^5).
	\end{aligned}\label{D.1}
	\end{equation}
	This leads to the transformation of $\varphi(\gamma(z))$ as follows:
	\begin{equation}
	\begin{aligned}
	\varphi(\gamma Z) &= \varphi\left(\gamma(z) -\frac{1}{2} \gamma'(z)\hspace{1mm}\frac{\overline{\gamma''(z)}}{\overline{\gamma'(z)}}\hspace{1mm}r^2+ \mathcal{O}(r^4)\right)
	\\
	& = \varphi(\gamma(z)) -\frac{1}{2} \gamma'(z)\hspace{1mm}\frac{\overline{\gamma''(z)}}{\overline{\gamma'(z)}}\hspace{1mm}r^2 \varphi'(\gamma(z)) + \mathcal{O}(r^4)\\
	&=\varphi(\gamma(z)) -\frac{1}{2} \frac{\overline{\gamma''(z)}}{\overline{\gamma'(z)}}\hspace{1mm}r^2 \frac{d}{dz}\left(\varphi(\gamma z)\right) + \mathcal{O}(r^4)\\
	&= \varphi(z) - \log |\gamma'(z)|^2 -\frac{1}{2} \frac{\overline{\gamma''(z)}}{\overline{\gamma'(z)}}\hspace{1mm} \left(\varphi'(z)-\frac{\gamma''(z)}{\gamma'(z)}\right)\hspace{1mm}r^2 + \mathcal{O}(r^4),
	\end{aligned}\label{D.2}
	\end{equation}
	where from the second line to the third line, the chain rule $
	\frac{1}{\gamma'(z)}\frac{d}{dz}\left(\varphi(\gamma z)\right) =  \varphi'(\gamma z)$
	and from the third line to the fourth line, the relation $\varphi\circ \gamma(z) = \varphi(z) - \log |\gamma'(z)|^2$ is used. For instance, according to \eqref{conicalasymp}, in the limit as  $z\rightarrow z_i$ near each conical singularity,  $\varphi(\gamma(z))$ has the following form:
	\begin{equation}
	\begin{aligned}
	&\varphi(\gamma(z))=-2\left(1-\frac{1}{m_i}\right) \log \left|\gamma(z)-\gamma(z_i)\right|+\log \frac{4\left|J_1^{(i)}(\gamma(z))\right|^{-\frac{2}{m_i}}}{m_i^2}+\smallO(1)\\
	&\hspace{1.28cm}=
	-2\left(1-\frac{1}{m_i}\right) \log \left|(z-z_i)\gamma'(z_i)\right|+\log \frac{4\left|J_1^{(i)}\gamma'(z)\right|^{-\frac{2}{m_i}}}{m_i^2}+\smallO(1)
	\\
	&\hspace{1.25cm}=\varphi(z)-\log\left|\gamma'(z)\right|^2.
	\end{aligned} 
	\end{equation}   
	The same conclusion can also be drawn near each puncture. Exponentiating \eqref{D.2}, we obtain:
	\begin{equation}
	e^{\varphi(\gamma Z)/2} = e^{\varphi(z)/2} \frac{1}{|\gamma'(z)|}\left(1-\frac{1}{4} \frac{\overline{\gamma''(z)}}{\overline{\gamma'(z)}}\hspace{1mm} \left(\varphi'(z)-\frac{\gamma''(z)}{\gamma'(z)}\right)\hspace{1mm}r^2 + \mathcal{O}(r^4)\right).\label{D.3}
	\end{equation}
	Consequently, from \eqref{D.1} and \eqref{D.3} for $\hat{f}$ in \eqref{fhat}, we find
	\begin{equation}
	\hat{f}(\gamma Z) 
	=re^{\varphi(z)/2}-\frac{1}{2} \frac{\overline{\gamma''(z)}}{\overline{\gamma'(z)}}\hspace{1mm} e^{\varphi(z)/2}\varphi'(z)\hspace{1mm}r^3+\mathcal{O}(r^5).
	\end{equation}
	Finally, plugging this into $f(Z)$ in \eqref{f},
	one finds 
	\begin{equation}
	f(Z) = re^{\varphi(z)/2}-\frac{1}{2}e^{\varphi(z)/2}\varphi'(z)\left(\sum_{\gamma \in \Gamma} \eta(\gamma Z)\frac{\overline{\gamma''(z)}}{\overline{\gamma'(z)}}\hspace{1mm} \right)r^3+\mathcal{O}(r^5).\label{fZ}
	\end{equation} 
	Note that the $\mathcal{O}(r^3)$ term depends on $\gamma$ through coefficients $c$ and $d$, complicating the summation over $\gamma$. However, only the leading term contributes to the renormalized volume calculation, so this dependency is not problematic.
	
	Before concluding this appendix, we introduce an alternative regularization method for addressing conical singularities, which extends the regularization used to obtain \eqref{fcompactt}. In the presence of conical singularity (with branching number $m_i$) at $z=z_i$ on the conformal boundary, the three-dimensional bulk metric, near each singularity, changes to\footnote{For example, see \cite{Krasnov_2001} for more details.}
	\begin{equation}
	\begin{aligned}
	ds^2_{\text{\tiny Conical}} = \frac{dr^2}{r^2} + \frac{1}{r^2} \frac{a_i^2}{(z-z_i)^{1-a_i}(\bar{z} - \bar{z}_i)^{1-a_i}}dzd\bar{z}, ~~~~~~~a_i=1-1/m_i
	\end{aligned}
	\label{conicalmetric}
	\end{equation}
	A viable generalization of the set of transformations \eqref{asyadstr}, which preserves the form of metric \eqref{conicalmetric} near each conical singularity, could easily be found to be
	\begin{equation}
	\begin{aligned}
	z(\gamma Z) &= \gamma(z)+\mathcal{O}(r^2),\\\bar{z}(\gamma Z) &= \overline{\gamma(z)}+\mathcal{O}(r^2),\\
	r(\gamma Z) &= |\gamma'|^{a_i}~r+ \mathcal{O}(r^3).
	\end{aligned}\label{conasyadstr}
	\end{equation}
	According to the asymptotic form of the field $\varphi$ near each singularity (see Appendix \ref{asymapp}), the automorphic function $f$ with respect to the transformation \eqref{conasyadstr} is given by
	\begin{equation}
	f(Z)= r~e^{\varphi(z)/2} +\mathcal{O}\left(r^3 \left(z-z_i\right)^{1/m_i -1}(\gamma'''(z_i)/3\gamma'(z_i)^{1-1/m_i})
	\right)~~~~~ \text{as} ~~r\rightarrow 0.
	\end{equation}
	At leading order, which serves a key role in determining the renormalized volume, the result matches that of the alternative regularization in \eqref{fZ}.    
	\bibliographystyle{jhep.bst}
	\bibliography{references}

\end{document}